\documentclass[
letterpaper,
twocolumn,
11pt,
aps,
accepted=2024-08-08
]{quantumarticle}
\pdfoutput=1
\usepackage[utf8]{inputenc}
\usepackage[english]{babel}
\usepackage[T1]{fontenc}

\usepackage[colorlinks=true,breaklinks=true,linkcolor=red,citecolor=magenta,urlcolor=magenta]{hyperref}

\usepackage{graphicx}
\usepackage{dcolumn}
\usepackage{bm}
\usepackage{dsfont}
\usepackage{amsthm}
\usepackage{chngcntr}
\usepackage{apptools}
\usepackage{amsmath}
\usepackage{mathtools}
\usepackage{todonotes}
\usepackage{graphicx}
\graphicspath{ {./images/} }
\usepackage{siunitx}
\usepackage[margin=1in,footskip=0.25in]{geometry}
\usepackage{blindtext}
\usepackage{enumitem}
\usepackage{xcolor}
\usepackage[makeroom]{cancel}

\usepackage{placeins}
\AtAppendix{\counterwithin{lemma}{section}}
\usepackage{longtable} 

\usepackage{algorithm}
\usepackage{algpseudocode}

\usepackage[capitalise]{cleveref}
\crefname{figure}{Fig.}{Fig.}

\newtheorem{theorem}{Theorem}
\newtheorem{corollary}{Corollary}[theorem]
\newtheorem{lemma}{Lemma}
\theoremstyle{definition}

\theoremstyle{theorem}

\theoremstyle{definition}

\newtheorem{remark}{Remark}

\usepackage{tikz}
\usetikzlibrary{quantikz}

\begin{document}

\title{Analysis of quantum Krylov algorithms with errors}

\author{William Kirby}
\email{william.kirby@ibm.com}
\affiliation{
IBM Quantum, IBM Research Cambridge, Cambridge, MA 02142, USA
}

\begin{abstract}

This work provides a nonasymptotic error analysis of quantum Krylov algorithms based on real-time evolutions, subject to generic errors in the outputs of the quantum circuits.
We prove upper and lower bounds on the resulting ground state energy estimates, and the error associated to the upper bound is linear in the input error rates.
This resolves a misalignment between known numerics, which exhibit approximately linear error scaling, and prior theoretical analysis, which only provably obtained scaling with the error rate to the power $\frac{2}{3}$.
Our main technique is to express generic errors in terms of an effective target Hamiltonian studied in an effective Krylov space.
These results provide a theoretical framework for understanding the main features of quantum Krylov errors.

\end{abstract}

\maketitle

\section{Introduction}
\label{intro}

In the last few years, quantum subspace diagonalization has emerged as a promising option for approximating ground state energies on quantum computers~\cite{mcclean2017subspace,colless2018computation,parrish2019filterdiagonalization,motta2020qite_qlanczos,takeshita2020subspace,huggins2020nonorthogonal,stair2020krylov,urbanek2020chemistry,kyriienko2020quantum,cohn2021filterdiagonalization,yoshioka2021virtualsubspace,epperly2021subspacediagonalization,seki2021powermethod,bespalova2021hamiltonian,cortes2022krylov,klymko2022realtime,jamet2022greens,baek2023nonorthogonal,lee2023sampling,zhang2023measurementefficient,kirby2023exactefficient,shen2023realtimekrylov,yang2023dualgse,yang2023shadow,ohkura2023leveraging,tkachenko2024davidson,yoshioka2024diagonalization,motta2023subspace}.
Quantum subspace diagonalization refers to quantum algorithms that calculate the projection of a target Hamiltonian into a low-dimensional subspace of the full Hilbert space.
The projection is then classically diagonalized to obtain the lowest energy in the subspace.
The accuracy of the resulting approximate ground state energy depends on the choice of subspace, for which many options have been proposed (see citations above).
For a recent review of quantum subspace methods, see~\cite{motta2023subspace}.

In this work, we focus on subspaces that are Krylov spaces, meaning that they are spanned by powers of some operator applied to a reference state~\cite{parrish2019filterdiagonalization,motta2020qite_qlanczos,stair2020krylov,urbanek2020chemistry,cohn2021filterdiagonalization,epperly2021subspacediagonalization,seki2021powermethod,bespalova2021hamiltonian,klymko2022realtime,jamet2022greens,lee2023sampling,kirby2023exactefficient,shen2023realtimekrylov,tkachenko2024davidson,yoshioka2024diagonalization}.
Even more specifically, we focus on Krylov spaces spanned by powers of a real-time evolution~\cite{parrish2019filterdiagonalization,stair2020krylov,urbanek2020chemistry,cohn2021filterdiagonalization,epperly2021subspacediagonalization,bespalova2021hamiltonian,klymko2022realtime,jamet2022greens,lee2023sampling,shen2023realtimekrylov,tkachenko2024davidson,yoshioka2024diagonalization}.
Real-time evolutions are natural to construct on a quantum computer, and such constructions have been extensively studied (e.g.,~\cite{lloyd1996quantumsimulators,childs2021trotter,aharonov2003adiabatic,childs2003quantumwalk,berry2007sparse,childs2010quantumwalk,childs2012lcu,berry2014exponentialimprovement,berry2015nearlyoptimal,low2017signalprocessing,hhkl2018,low2018interactionpicture,tran2019locality,low2019qubitization,berry2020timedependent}).
We refer to the corresponding subspace diagonalization method as the \emph{quantum Krylov algorithm}; this method has been demonstrated experimentally on up to 56 qubits~\cite{yoshioka2024diagonalization}.

These real-time Krylov spaces are also advantageous because the resulting approximate ground state energies possess analytic convergence bounds~\cite{epperly2021subspacediagonalization}, like the classical Krylov space spanned by powers of the target Hamiltonian~\cite{paige1971computation,kaniel1966linearalgebra,saad1980lanczos}.
The seminal error analysis in~\cite{epperly2021subspacediagonalization} also showed that error bounds can be obtained even in the presence of noise.
However, these error bounds rely on a number of quite stringent assumptions, and have some suboptimal features: in particular, sublinear dependence of the energy error on the noise rate.
In the present work, we resolve some of the less desirable properties of the analysis in~\cite{epperly2021subspacediagonalization}, in particular requiring only a few, relatively weak assumptions on the noise, and subject to these, obtaining linear dependence on noise rate for the upper bound.
The lower bound is weaker but can be tightened in a tradeoff with the upper bound.

\subsection{Real-time quantum Krylov algorithm}

As noted above, the real-time quantum Krylov algorithm is based on projecting a Hamiltonian $H$ of interest into the Krylov space spanned by its real-time evolutions applied to some initial reference state.
This nonorthogonal basis may be written
\begin{equation}
\label{exact_krylov_space}
    \textbf{V}=[e^{-idH\,dt}|\psi_0\rangle,~e^{-i(d-1)H\,dt}|\psi_0\rangle,...,~e^{idH\,dt}|\psi_0\rangle]
\end{equation}
for some timestep $dt$, and total Krylov dimension
\begin{equation}
    D=2d+1.
\end{equation}
In the ideal case, we find the lowest energy in this subspace by finding the least eigenvalue of
\begin{equation}
\label{mat_pencil}
    (\textbf{H},\textbf{S})=(\textbf{V}^\dagger H\textbf{V}, \textbf{V}^\dagger\textbf{V}),
\end{equation}
i.e., by solving
\begin{equation}
\label{gep}
    \textbf{H}v=\lambda\textbf{S}v.
\end{equation}
Note that $\textbf{H},\textbf{S}$ are $D\times D$ Hermitian matrices and $\textbf{S}$ is positive semidefinite.
In practice, $(\textbf{H},\textbf{S})$ are evaluated elementwise using
\begin{equation}
\label{ideal_matrix_elements}
\begin{split}
    \textbf{H}_{jk}&=\langle\psi_0|e^{-ijH\,dt}He^{ikH\,dt}|\psi_0\rangle,\\
    \textbf{S}_{jk}&=\langle\psi_0|e^{-ijH\,dt}e^{ikH\,dt}|\psi_0\rangle,
\end{split}
\end{equation}
followed by solving the generalized eigenvalue problem \eqref{gep}.
Note that $i$ is the imaginary number throughout, and is never used as an index.

Solving the generalized eigenvalue problem \eqref{gep} requires regularizing it, since in practice $\textbf{S}$ is nearly singular, meaning that under errors due to noise it in general is ill-conditioned and may even fail to be positive semidefinite.
Previous works have suggested regularizing the generalized eigenvalue problem by projecting out dimensions corresponding to eigenvectors of $\textbf{S}$ with eigenvalues smaller than some threshold $\epsilon>0$~\cite{klymko2022realtime,epperly2021subspacediagonalization,kirby2023exactefficient}. 
We focus on this \emph{thresholding} technique here.
Another option is adding small multiples of the identity to $\textbf{S}$ (and possibly also $\textbf{H}$)~\cite{zhang2023measurementefficient}, which is a form of Tikhonov regularization.

\subsection{Motivation and main results}
\label{main_results}

A typical qualitative argument that is made to explain the noise resilience of quantum Krylov algorithms is that the noise just disturbs the Krylov space, and we still find the lowest energy in that noisy subspace.
This is true if the time evolutions in \eqref{ideal_matrix_elements} are subject to some errors like Trotter approximation, provided the matrix elements are evaluated exactly as written in \eqref{ideal_matrix_elements}.
In this case, one will still obtain valid (i.e., variational) energies.
This is a potential advantage of the quantum Krylov algorithm compared to other ground state energy estimation methods, particularly those based on extracting energy eigenvalues purely from time evolutions.

However, the above argument fails if $(\textbf{H},\textbf{S})$ are subject to some more general errors that cannot be expressed as just a disturbance of the Krylov space.
Examples of this include finite sample noise or other constructions than the exact one shown in \eqref{ideal_matrix_elements}, when the time-evolutions are approximate.
The failure of this ``disturbance of the Krylov space'' argument provided the inspiration for the main idea in this work.
We show in \cref{noise_as_perturbation} that generic errors can be modeled as a disturbance of the Krylov space together with a disturbance of the Hamiltonian itself.
This disturbance yields an effective Hamiltonian corresponding to the noisy matrix pair.
Given this, there is a path towards bounding energy errors in the general case using an argument similar to the one above, but now additionally accounting for the error in the effective Hamiltonian.

We pursue this goal in \cref{lower_bound_sec,upper_bound_sec}, which build up to lower and upper bounds (respectively) on the signed energy error. 
The lower bound serves to limit the amount of violation of variationality that is possible in a noisy quantum Krylov algorithm.
However, it is the weaker of the two results, since in order to be useful it requires a larger choice of regularization threshold than is typically found to be optimal.
We include it in spite of this as an illustration of the application of the technique of modeling error via an effective Hamiltonian and Krylov space, and also as motivation for future work.

The main result is the upper bound in \cref{complete_bound_thm}, which is the culmination of \cref{upper_bound_sec}.
A useful special case is given in \eqref{approx_proj_bound_gs}, which corresponds to a particular choice of the free parameters in \cref{complete_bound_thm}.
The notable features of \cref{complete_bound_thm} are:
\begin{enumerate}
    \item The bound is linear in the input error rate. The theorem is nonasymptotic, but we give a simplified asymptotic version of it below. Let $\eta$ be the spectral norm of the errors in $\textbf{H},\textbf{S}$, let $|\gamma_0|^2$ be the initial state's probabilistic overlap with the true ground state, and let $\Delta$ be the spectral gap of $H$.
    Then the signed error in the ground state energy estimate is upper bounded as
    \begin{equation}
    \label{asymptotic_bound}
        \hspace{-0.15in}\text{energy error}\le O\left(\frac{\left(\frac{1}{\Delta}+D\right)\eta+(1+\beta)^{-2d}}{|\gamma_0'|^2}\right),
    \end{equation}
    where $\beta=\Theta(\Delta)$.
    Note that the explicit dependence of \eqref{asymptotic_bound} on the effective spectral gap can be removed, but we use the special case in \eqref{asymptotic_bound} for simplicity.
    The main point is that \eqref{asymptotic_bound} is linear in the noise rate $\eta$.
    
    \item The assumptions required for \cref{complete_bound_thm} are extremely weak. Essentially they state that...
    \begin{enumerate}
        \item the initial state overlap $|\gamma_0|^2$ must be sufficiently large that it is not overwhelmed by either the noise or the thresholding procedure;
        \item the theorem does not hold in a particular regime where the bound would be larger than the Hamiltonian norm anyway.
    \end{enumerate}
\end{enumerate}

The elephant in the room is of course that this error bound is only one-sided, and as mentioned previously, the lower bound is weak when the threshold is chosen $\epsilon=O(\eta)$ as above.
In practice, thresholds $\epsilon=O(\eta)$ are not typically found to lead to large negative fluctuations and violations of variationality (see~\cite{klymko2022realtime,kirby2023exactefficient} and \cref{numerics}), so one may hope that the lower bound can be improved.
However, even the present bound can be tightened by choosing $\epsilon=O(\sqrt{\eta})$: in this case, \cref{lower_bound_thm} yields a lower bound that is $O(\sqrt{\eta})$, and \cref{complete_bound_thm} also yields an upper bound that is $O(\sqrt{\eta})$.
This would represent a more conservative approach to thresholding a noisy Krylov algorithm, sacrificing the tighter upper bound in order to limit violations of variationality.
In practice, one might be able to use heuristics to detect fluctuations due to ill-conditioning, and choose the optimal threshold in this way.

\section{Noise as error in subspace and Hamiltonian}
\label{noise_as_perturbation}

Suppose $(\textbf{H},\textbf{S})$ are calculated with errors, yielding some faulty matrices $(\textbf{H}',\textbf{S}')$.
One motivating example is
\begin{equation}
\label{approx_matrix_elements}
\begin{split}
    \textbf{H}_{jk}'&=\langle\psi_0|\text{PF}(k-j)H|\psi_0\rangle,\\
    \textbf{S}_{jk}'&=\langle\psi_0|\text{PF}(k-j)|\psi_0\rangle,
\end{split}
\end{equation}
where $\text{PF}(k-j)$ is a product formula approximation to $e^{i(k-j)H\,dt}$.
Note that if $\text{PF}(k-j)$ were the exact time evolution that it approximates, then \eqref{approx_matrix_elements} would be equivalent to \eqref{ideal_matrix_elements}, since exact time evolutions commute with $H$.
However, once the time evolutions are approximated, the two expressions are no longer equivalent.
Another unavoidable source of error is estimation of the above matrix elements with a finite number of samples.
In the analysis that follows, we do not assume any particular source for the errors, merely quantifying them as $\|\textbf{H}'-\textbf{H}\|$ and $\|\textbf{S}'-\textbf{S}\|$, and obtaining bounds in terms of these.

Our main technique in this paper is to express the noisy matrix pair $(\textbf{H}',\textbf{S}')$ in terms of an effective Hamiltonian $H'$ and an effective Krylov basis, whose vectors form the columns of a matrix $\textbf{V}'$.
Naively, one might hope to write
\begin{equation}
\label{noisy_mat_pencil}
    (\textbf{H}',\textbf{S}')=(\textbf{V}'^\dagger H'\textbf{V}', \textbf{V}'^\dagger\textbf{V}'),
\end{equation}
in direct analogy to \eqref{mat_pencil}.
However, an immediate obstacle is that the faulty overlap matrix $\textbf{S}'$ may not be positive semidefinite (p.s.d.), which is a necessary and sufficient condition for representing it in the form in \eqref{noisy_mat_pencil}.
Only in some special cases is $\textbf{S}'$ guaranteed to be p.s.d., e.g., if $\textbf{S}'$ is constructed as in \eqref{approx_matrix_elements} and
\begin{equation}
    \text{PF}(k-j)=(U')^{k-j}
\end{equation}
for some $U'$.
This is the case where, for each $k-j$, $\text{PF}(k-j)$ is obtained as $k-j$ repetitions of some fixed step $U'$, and consequently we could take the Krylov vectors to be $(U')^j|\psi_0\rangle$.

However, a generic $\textbf{S}'$ constructed as in \eqref{approx_matrix_elements}, such as when $\text{PF}(k-j)$ is obtained by a fixed number of Trotter steps whose evolution times scale with $k-j$, is not guaranteed to be p.s.d. and in general turns out not to be.
The effects of finite sample and device noise additionally do not preserve positive semidefiniteness.
Let us still assume at least that $\textbf{S}'$ is Hermitian.
Even if $\text{PF}(j-k)\neq \text{PF}(k-j)^\dagger$, Hermitianity of $\textbf{S}'$ can be enforced by only using \eqref{approx_matrix_elements} to calculate the matrix elements on and above the diagonal, and obtaining the below-diagonal matrix elements as conjugates of their transposes.

A non-p.s.d.~$\textbf{S}'$ cannot be expressed in the form $\textbf{V}'^\dagger\textbf{V}'$.
However, in order to solve the generalized eigenvalue problem \eqref{gep} for the noisy matrix pair $(\textbf{H}',\textbf{S}')$, we will have to regularize the problem, since it cannot be solved numerically unless $\textbf{S}'$ is well-conditioned.
As discussed in \cref{intro}, we will accomplish this by removing eigenspaces of $\textbf{S}'$ whose eigenvalues lie below some threshold $\epsilon>0$, from both $\textbf{H}'$ and $\textbf{S}'$.
We will then solve the generalized eigenvalue problem for the resulting (lower dimensional) pair $\big(\widetilde{\textbf{H}}',\widetilde{\textbf{S}}'\big)$.
Below, we will refer to this process as ``thresholding at $\epsilon$.''
We will discuss the details of this and the resulting analysis below.
For now we observe that it means that we will ultimately be solving a generalized eigenvalue problem with an overlap matrix $\widetilde{\textbf{S}}'$ that is positive definite with least eigenvalue lower bounded by $\epsilon$, the regularization threshold.

For the purpose of the analysis, it is convenient to introduce yet another intermediate matrix $\textbf{S}''$, which is obtained from $\textbf{S}'$ by replacing all eigenvalues of $\textbf{S}'$ that are below $\epsilon$ with $0$, preserving the same eigenvectors.
This is useful because $\textbf{S}''$ is p.s.d.~by construction, but it still has the same dimensions as $\textbf{S}'$, and the lower dimensional, thresholded overlap matrix $\widetilde{\textbf{S}}'$ could be obtained by removing the null space of $\textbf{S}''$.
Since $\textbf{S}''$ is p.s.d., we can express it as $\textbf{V}'^\dagger\textbf{V}'$ for some $\textbf{V}'$.

We can define $\textbf{S}''$ formally by letting $\Pi'$ denote the projector onto eigenspaces of $\textbf{S}'$ with eigenvalues above $\epsilon$: then
\begin{equation}
\label{Sprimeprime_def}
    \textbf{S}''\coloneqq\Pi'\textbf{S}'\Pi'.
\end{equation}
Since thresholding also requires projecting the corresponding dimensions out of $\textbf{H}'$, we also introduce
\begin{equation}
\label{Hprimeprime_def}
    \textbf{H}''\coloneqq\Pi'\textbf{H}'\Pi'.
\end{equation}
This new matrix pair $(\textbf{H}'',\textbf{S}'')$ is equal to the noisy, thresholded matrix pair $\big(\widetilde{\textbf{H}}',\widetilde{\textbf{S}}'\big)$ up to padding by extra dimensions (all zeroes, corresponding to the dimensions projected out by $\Pi'$).
Hence their energies $\widetilde{E}_i$ are the same, although $(\textbf{H}'',\textbf{S}'')$ is a purely theoretical construction whose generalized eigenvalue problem could not actually be solved in practice.\footnote{Since $\textbf{S}''$ is singular unless no eigenvalues of $\textbf{S}'$ are below $\epsilon$, i.e., no dimensions are removed by thresholding.}
These are the energies that would come out of the noisy, thresholded quantum algorithm.

In the analysis that follows, we will obtain both upper and lower bounds on the energy estimates resulting from a Krylov matrix pair with errors.
We will twice have the opportunity to illustrate the type of scheme described above, i.e., representing the errors in terms of an effective Krylov basis and Hamiltonian, since we will use different effective bases and Hamiltonians for the lower and upper bounds.
More broadly, one may hope that this approach to analyzing matrix pairs can be useful in other contexts.

\section{Lower bound on energy error}
\label{lower_bound_sec}

\subsection{Effective Krylov space and Hamiltonian}
\label{effective_for_lower}

As discussed in \cref{noise_as_perturbation}, we begin by expressing the noisy matrix pair $(\textbf{H}',\textbf{S}')$ in terms of an effective Krylov basis and an effective Hamiltonian.
To construct an effective Krylov basis whose overlap matrix is $\textbf{S}''$, we can begin by diagonalizing both $\textbf{S}'$ and the original overlap matrix $\textbf{S}$, via unitaries $Q'$ and $Q$, respectively.
Note that to execute the quantum Krylov algorithm we do not actually need to perform this diagonalization in practice, which would be impossible for the unknown ideal matrix $\textbf{S}$.
We are only concerned with demonstrating existence of an effective Krylov basis with certain properties.

Let
\begin{equation}
\label{diagonalized_Ss}
\begin{split}
    \Lambda&\coloneqq Q^\dagger\textbf{S}Q=\text{diag}(\lambda_0,\lambda_1,...,\lambda_{D-1}),\\
    \Lambda'&\coloneqq Q'^\dagger\textbf{S}'Q'=\text{diag}(\lambda_0',\lambda_1',...,\lambda_{D-1}'),
\end{split}
\end{equation}
where $\lambda_i,\lambda_i'$ are the eigenvalues of $\textbf{S},\textbf{S}'$, respectively, in weakly increasing order.
By definition \eqref{Sprimeprime_def}, $Q$ also diagonalizes $\textbf{S}''$; define $\Lambda''$ to be the corresponding diagonal matrix of eigenvalues.

Next, let $\sqrt{\textbf{S}''}$ and $\sqrt{\textbf{S}}$ denote the Hermitian square-roots of $\textbf{S}''$ and $\textbf{S}$, i.e.,
\begin{equation}
\begin{split}
    &\sqrt{\textbf{S}}=Q\sqrt{\Lambda}Q^\dagger,\\
    &\sqrt{\textbf{S}''}=Q'\sqrt{\Lambda''}Q'^\dagger.
\end{split}
\end{equation}
Since $\textbf{S}$ is the Gram matrix of $\textbf{V}$, the polar decomposition of $\textbf{V}$ is
\begin{equation}
\label{V_polar}
    \textbf{V}=F\sqrt{\textbf{S}}
\end{equation}
for some matrix $F$ with orthonormal columns.
This implies that $F^\dagger F=\mathds{1}_{D\times D}$, so if we define our effective Krylov basis $\textbf{V}'$ to be
\begin{equation}
\label{Vprime_def_lower}
    \textbf{V}'\coloneqq FG\sqrt{\textbf{S}''}
\end{equation}
for any $D\times D$ unitary $G$, then
\begin{equation}
\label{Vprimeprime_condition}
    \textbf{V}'^\dagger\textbf{V}'=\sqrt{\textbf{S}''}G^\dagger F^\dagger FG\sqrt{\textbf{S}''}=\textbf{S}''
\end{equation}
as desired.
We leave $G$ arbitrary for now.

We now move on to $\textbf{H}''$.
The matrix $\textbf{V}'$ of effective Krylov vectors forms a $D$-dimensional coordinate system.
We want an effective Hamiltonian $H'$ whose block in the subspace spanned by $\textbf{V}'$ is $\textbf{H}''$, i.e., we require that $\textbf{V}'^\dagger H'\textbf{V}'=\textbf{H}''$.
The remainder of $H'$ we can take to be equal to the corresponding part of $H$, since it is outside of the Krylov space and hence will play no role in our calculations.

A corresponding expression for $H'$ is
\begin{equation}
\label{Hprime_def}
    H'=H+\textbf{V}'\textbf{S}''^+\left(\textbf{H}'-\textbf{V}'^\dagger H\textbf{V}'\right)\textbf{S}''^+\textbf{V}'^\dagger,
\end{equation}
where $\textbf{S}''^+$ denotes the Moore-Penrose pseudoinverse of $\textbf{S}''$.
To see that this expression yields the desired relation $\textbf{V}'^\dagger H'\textbf{V}'=\textbf{H}''$, we conjugate \eqref{Hprime_def} by $\textbf{V}'$:
\begin{equation}
\begin{split}
    \textbf{V}'^\dagger H'\textbf{V}'&=\textbf{V}'^\dagger H\textbf{V}'+\textbf{S}''\textbf{S}''^+\left(\textbf{H}'-\textbf{V}'^\dagger H\textbf{V}'\right)\textbf{S}''^+\textbf{S}''\\
    &=\textbf{V}'^\dagger H\textbf{V}'+\Pi'\left(\textbf{H}'-\textbf{V}'^\dagger H\textbf{V}'\right)\Pi'\\
    &=\Pi'\textbf{H}'\Pi'\\
    &=\textbf{H}'',
\end{split}
\end{equation}
where the first line uses \eqref{Vprimeprime_condition}, the second line follows because ${\textbf{S}''^+\textbf{S}''=\textbf{S}''\textbf{S}''^+=\Pi'}$, the third line follows because $\textbf{V}'\Pi'=\textbf{V}'$ (by \eqref{Vprime_def_lower} and the fact that $\sqrt{\textbf{S}''}\Pi'=\sqrt{\textbf{S}''}$), and the last line follows by \eqref{Hprimeprime_def}.
Hence, with the above choices of $\textbf{V}'$ and $H'$, we have
\begin{equation}
\label{perturbed_mat_pencil}
    (\textbf{H}'',\textbf{S}'')=(\textbf{V}'^\dagger H'\textbf{V}', \textbf{V}'^\dagger\textbf{V}'),
\end{equation}
which has the same spectrum as the thresholded problem $(\widetilde{\textbf{H}}',\widetilde{\textbf{S}}')$.
Some caution is required because, as noted above, the matrix pair $(\textbf{H}'',\textbf{S}'')$ is singular: by ``has the same spectrum as the thresholded problem,'' we mean that the well-defined energies of $(\textbf{H}'',\textbf{S}'')$ are equal to the spectrum of the thresholded problem.
With this understood, we can think of the thresholded problem as studying the effective Hamiltonian $H'$ in the effective Krylov space $\text{span}(\textbf{V}')$.

We now want to bound the difference between the effective Hamiltonian $H'$ and exact Hamiltonian $H$:
\begin{theorem}
\label{Hprime_diff_thm_lower}
    Let the unitary $G$ in the definition \eqref{Vprime_def_lower} of $\textbf{V}\,'$ be defined such that
    \begin{equation}
    \label{G_def}
        \sqrt{\textbf{S}}\,\Pi'=G\sqrt{\Pi'\textbf{S}\,\Pi'}
    \end{equation}
    is the polar decomposition of $\sqrt{\textbf{S}}\,\Pi'$.
    Assume that $\|\textbf{S}\,'-\textbf{S}\,\|\le\epsilon$.
    Then for $H'$ as defined in \eqref{Hprime_def},
    \begin{equation}
    \label{Hprime_diff_lower}
        \|H'-H\|\le\frac{\|\textbf{H}\,'-\textbf{H}\,\|+(1+\sqrt{2})\|\textbf{S}\,'-\textbf{S}\,\|\|H\|}{\epsilon}.
    \end{equation}
\end{theorem}
\noindent
The proof is given in \cref{proofs_app}.
An explanation of why \eqref{G_def} is a valid polar decomposition is given in the proof.

\subsection{Lower bound}

The lower bound on the energy error from the noisy, thresholded problem follows immediately as a corollary of \cref{Hprime_diff_thm_lower}:
\begin{corollary}
\label{lower_bound_thm}
    Let $H$ be a Hamiltonian, let $(\textbf{H},\textbf{S}\,)=(\textbf{V}\,^\dagger H\textbf{V}, \textbf{V}\,^\dagger\textbf{V}\,)$ be a real-time Krylov matrix pair representing $H$ in the Krylov space $\text{span}(\textbf{V}\,)$, and let $(\textbf{H}\,',\textbf{S}\,')$ be a Hermitian approximation to $(\textbf{H},\textbf{S}\,)$.
    Let $E_0$ be the ground state energy of $H$, which we want to estimate.
    Then the energy error of lowest energy $\widetilde{E}_0$ of $(\textbf{H}\,',\textbf{S}\,')$ after thresholding at $\epsilon$ is lower bounded as
    \begin{equation}
    \label{approx_lower_bound}
        \widetilde{E}_0-E_0\ge-\frac{\|\textbf{H}\,'-\textbf{H}\,\|+(1+\sqrt{2})\|\textbf{S}\,'-\textbf{S}\,\|\|H\|}{\epsilon}.
    \end{equation} 
\end{corollary}
\begin{proof}
    By Weyl's theorem (\cite{bhatia1997matrix}, Cor. III.2.6; see also \cref{weyl_thm} in \cref{proofs_app}), the difference between the lowest eigenvalues of $H'$ and $H$ is upper bounded by $\|H'-H\|$, which is in turn upper bounded as in \eqref{Hprime_diff_lower}.
    By the Rayleigh-Ritz variational principle, since $H'$ is the Hamiltonian corresponding to the matrix pair $(\textbf{H}'',\textbf{S}'')$, as in \eqref{perturbed_mat_pencil}, the lowest energy of $(\textbf{H}'',\textbf{S}'')$ is lower bounded by the lowest energy of $H'$.
    Finally, as explained after \eqref{perturbed_mat_pencil}, the energies of $(\textbf{H}'',\textbf{S}'')$ are the same as the energies of the noisy, thresholded problem.
    The result follows.
\end{proof}
\begin{remark}
    The error bound in \eqref{approx_lower_bound} is weak compared to numerical results, since in practice, it is typically found that the optimal threshold $\epsilon$ is of the same order as the noise rates $\|\textbf{H}\,'-\textbf{H}\|$ and $\|\textbf{S}\,'-\textbf{S}\|\|H\|$ (see also \cref{error_bounds_sec,numerics}).
    However, we have included it since it illustrates the use of the effective Krylov space and Hamiltonian technique, and also as a suggestion for future work to improve the bound.
    See \cref{numerics} for a detailed discussion.
\end{remark}

\section{Upper bound on energy error}
\label{upper_bound_sec}

\subsection{Effective Krylov space and Hamiltonian}
\label{effective_for_upper}

For the energy error upper bound, we will be considering a particular point in the Krylov space, whose energy will be the upper bound we want to obtain.
For now, we can give that point a generic label $c'$, a $D$-dimensional coordinate vector whose corresponding state in the effective Krylov space will be $\textbf{V}'c'$.
The only constraint we will put on $c'$ for now is that it lives in the range of $\textbf{S}''$, i.e., the subspace spanned by eigenvectors of $\textbf{S}'$ with eigenvalues above threshold.
We can formalize this by requiring that
\begin{equation}
\label{cprime_req}
    c'=\Pi'c'.
\end{equation}
We also require that $c'$ is nonzero.

The key observation is that since we will only consider the point $c'$, the effective Krylov space we are about to construct only needs to match $\textbf{S}''$ at that point.
In other words, if we denote the effective Krylov basis as $\textbf{V}'$, we require
\begin{equation}
\label{vprime_req}
    c'^\dagger\textbf{V}'^\dagger\textbf{V}'c'=c'^\dagger\textbf{S}''c'=c'^\dagger\textbf{S}'c',
\end{equation}
where the second equality follows by \eqref{cprime_req} and the definition \eqref{Sprimeprime_def} of $\textbf{S}''$.
Even though $\textbf{S}''$ has least eigenvalue zero, by \eqref{cprime_req} $c'^\dagger\textbf{S}''c'\neq0$ as long as $c'\neq0$, which we assumed.
Given this, a convenient choice of $\textbf{V}'$ that satisfies \eqref{vprime_req} is
\begin{equation}
\label{Vprime_def}
    \textbf{V}'\coloneqq\sqrt{\frac{c'^\dagger\textbf{S}'c'}{c'^\dagger\textbf{S}c'}}\,\textbf{V}=\sqrt{\frac{c'^\dagger\textbf{S}''c'}{c'^\dagger\textbf{S}c'}}\,\textbf{V},
\end{equation}
i.e., we choose the effective Krylov space to simply be whatever rescaling of the ideal Krylov space yields the correct length for the vector $\textbf{V}'c'$.

The second component we want is an effective Hamiltonian $H'$ that yields the noisy, thresholded Krylov matrix $\textbf{H}''$ with respect to our effective Krylov space $\text{span}(\textbf{V}')$.
Again, we only require this at the point $c'$, so we assert that
\begin{equation}
\label{Hprime_req}
    c'^\dagger\textbf{V}'^\dagger H'\textbf{V}'c'=c'^\dagger\textbf{H}''c'=c'^\dagger\textbf{H}'c',
\end{equation}
where just as in \eqref{vprime_req}, the second equality follows by \eqref{cprime_req} and the definition \eqref{Hprimeprime_def} of $\textbf{H}''$.
At all points besides $\textbf{V}'c'$, we are free to choose the value of $H'$, so we can let it be equal to $H$ elsewhere.
This yields the following form for $H'$: with $|\psi\rangle\coloneqq\textbf{V}'c'$,
\begin{equation}
\label{Hprime_forms}
\begin{split}
    H'&\coloneqq H+\left(c'^\dagger\textbf{H}''c'-c'^\dagger\textbf{V}'^\dagger H\textbf{V}'c'\right)\frac{|\psi\rangle\langle\psi|}{\langle\psi|\psi\rangle^2}\\
    &=H+\left(c'^\dagger\textbf{H}'c'-c'^\dagger\textbf{V}'^\dagger H\textbf{V}'c'\right)\frac{|\psi\rangle\langle\psi|}{\langle\psi|\psi\rangle^2}.
\end{split}
\end{equation}
Taking the expectation value of both sides with respect to $|\psi\rangle$ verifies that this form satisfies \eqref{Hprime_req}.

The final question to answer in this section is: how far are these effective objects from their ideal counterparts?
The distance ${\|\textbf{V}'-\textbf{V}\|}$ will not turn out to matter to us directly, but the distance ${\|H'-H\|}$ will, and we can bound it as follows: using the second line of \eqref{Hprime_forms} we obtain
\begin{equation}
\label{Hprime_diff}
\begin{split}
    &\|H'-H\|=\left|c'^\dagger\textbf{H}'c'-c'^\dagger\textbf{V}'^\dagger H\textbf{V}'c'\right|\left\|\frac{|\psi\rangle\langle\psi|}{\langle\psi|\psi\rangle^2}\right\|\\
    &=\frac{\left|c'^\dagger\textbf{H}'c'-c'^\dagger\textbf{V}'^\dagger H\textbf{V}'c'\right|}{\langle\psi|\psi\rangle}\\
    &\le\frac{\left|c'^\dagger\textbf{H}'c'-c'^\dagger\textbf{H}c'\right|+\left|c'^\dagger\textbf{V}^\dagger H\textbf{V}c'-c'^\dagger\textbf{V}'^\dagger H\textbf{V}'c'\right|}{\langle\psi|\psi\rangle}\\
    &\le\frac{\|c'\|^2\left\|\textbf{H}'-\textbf{H}\right\|+\left|c'^\dagger\textbf{V}^\dagger H\textbf{V}c'\left(1-\frac{c'^\dagger\textbf{S}'c'}{c'^\dagger\textbf{S}c'}\right)\right|}{\langle\psi|\psi\rangle}\\
    &\le\|c'\|^2\frac{\left\|\textbf{H}'-\textbf{H}\right\|+\|H\|\left\|\textbf{S}-\textbf{S}'\right\|}{\langle\psi|\psi\rangle},
\end{split}
\end{equation}
where the third step uses $\textbf{H}=\textbf{V}^\dagger H\textbf{V}$, the fourth step follows by inserting \eqref{Vprime_def}, and the final step follows because
\begin{equation}
    \frac{|c'^\dagger\textbf{V}^\dagger H\textbf{V}c'|}{c'^\dagger\textbf{S}c'}\le\|H\|
\end{equation}
by the Rayleigh-Ritz variational principle.

\subsection{Upper bound}
\label{error_bounds_sec}

\cref{effective_for_upper} showed that we can think of the noisy, thresholded problem $(\widetilde{\textbf{H}}\,',\widetilde{\textbf{S}}\,')$ as studying the effective Hamiltonian $H'$ in the effective Krylov space $\text{span}(\textbf{V}')$, at least at the point $c'$ (which we have yet to specify).
We now begin to work our way towards an upper bound on $\widetilde{E}_0$, the lowest energy of that noisy, thresholded problem.
First, we show that the effective Krylov space span$(\textbf{V}')$ contains an approximate ground state of the exact Hamiltonian $H$:
\begin{theorem}[partly derived from Theorem 3.1 in~\cite{epperly2021subspacediagonalization}]
\label{epperly_thm}
    Let $d$ be a positive integer defining the dimension $D=2d+1$ as above, let $\delta>0$, let $(E_k,|E_k\rangle)$ be the eigenpairs of $H$ in weakly increasing order of energy, and let $R\coloneqq E_\text{max}-E_0$ be the spectral range of $H$.
    Let
    \begin{equation}
        |\psi_0\rangle=\sum_{k=0}^{N-1}\gamma_k|E_k\rangle
    \end{equation}
    be the expansion of $|\psi_0\rangle$ in the energy eigenbasis of $H$, where $N$ is the Hilbert space dimension.
    Assume
    \begin{equation}
    \label{thm1_assumption}
        \|\textbf{S}\,'-\textbf{S}\,\|\le\epsilon.
    \end{equation}
    Then there exists an operator $P$ such that the column space of $\textbf{V}\,'$ contains a state
    \begin{equation}
    \label{ansatz_approx}
        |\psi\rangle=P|\psi_0\rangle
    \end{equation}
    and $P$ satisfies
    \begin{equation}
    \label{agsp_def_part_1}
        P|E_k\rangle=\beta'_k|E_k\rangle,
    \end{equation}
    where
    \begin{equation}
    \label{agsp_def_part_2}
        |\beta'_k|^2\le
        \begin{cases}
            2+\alpha_k\quad\text{if $E_k-E_0<\delta$},\\
            8\left(1+\frac{\pi\delta}{R}\right)^{-2d}+\alpha_k\quad\text{if $E_k-E_0\ge\delta$}.
        \end{cases}
    \end{equation}
    The $\alpha_k$ satisfy
    \begin{equation}
    \label{agsp_def_part_3}
        \sum_{k=0}^{N-1}|\gamma_k|^2\alpha_k\le2D\left(\epsilon+\|\textbf{S}\,'-\textbf{S}\,\|\right).
    \end{equation}
    The norm of $|\psi\rangle$ is can be lower bounded with or without explicit dependence on $c'$, the coordinates of $|\psi\rangle$ in the column space of $\textbf{V}\,'$:
    \begin{equation}
    \label{psi_norm_lower_bound}
    \begin{split}
        &\||\psi\rangle\|^2\ge\|c'\|^2\left(|\gamma_0|^2-\epsilon-\|\textbf{S}\,'-\textbf{S}\,\|\right),\\
        &\||\psi\rangle\|^2\ge|\gamma_0|^2-2\epsilon-2\|\textbf{S}\,'-\textbf{S}\,\|.
    \end{split}
    \end{equation}
\end{theorem}
~\\
\noindent
We give a proof in \cref{proofs_app} because there is a significant difference from the proof of Theorem 3.1 in~\cite{epperly2021subspacediagonalization}: in~\cite{epperly2021subspacediagonalization}, the same ansatz coordinates are used in the faulty Krylov space as in the ideal Krylov space, but in this work we modify the ansatz coordinates (one can already see that this will be necessary given \eqref{cprime_req}), which leads to a more complex dependence involving the difference $\|\textbf{S}'-\textbf{S}\|$.
We point out the details in the proof.
However, it is important to acknowledge that the other main ideas in \cref{epperly_thm}, specifically the choice of ideal ansatz, are based on~\cite{epperly2021subspacediagonalization}.

As for interpretation of \cref{epperly_thm}, $|\psi\rangle$ is an approximate ground state of $H$ because it comes from application of the approximate ground state projector $P$ of $H$ to the initial state $|\psi_0\rangle$.
To see that $P$ is an approximate ground state projector of $H$, note that \eqref{agsp_def_part_1} and \eqref{agsp_def_part_2} show that $P$ suppresses amplitudes of energy eigenstates of $H$ with energies above $E_0+\delta$ by the exponentially-vanishing factor $8\left(1+\frac{\pi\delta}{R}\right)^{-2d}$ plus the additional term $\alpha$ due to noise in the effective Krylov space.
Also key is \eqref{psi_norm_lower_bound}, which guarantees that $P$ is not just suppressing the entire state, since the total norm of $|\psi\rangle$ is lower bounded.

Next, we show that the error in the ground state energy estimate obtained by taking the expectation value of some other Hamiltonian $H'$ with respect to the approximate ground state $|\psi\rangle$ constructed in \cref{epperly_thm} is upper bounded, with the bound depending $\|H'-H\|$:
\begin{widetext}
\begin{theorem}
\label{approx_projector_perturbed_H_thm}
    Let $H$ and $H'$ be Hamiltonians.
    Let $E_0$ be the ground state energy of $H$, which we want to estimate.
    Let
    \begin{equation}
    \label{psi_def}
        |\psi\rangle=P|\psi_0\rangle
    \end{equation}
    be the approximately projected state defined in the statement of \cref{epperly_thm}.
    Let
    \begin{equation}
    \label{deltatilde_def}
        \widetilde{\Delta}\coloneqq E_1'-E_0
    \end{equation}
    be the gap between the ground state energy of $H$ and the first excited energy of $H'$, and let $\mathbf{1}$ denote the indicator function, i.e.,
    \begin{equation}
        \mathbf{1}(\delta'>\widetilde{\Delta})=
        \begin{cases}
            1\quad\text{if $\delta'>\widetilde{\Delta}$},\\
            0\quad\text{if $\delta'\le\widetilde{\Delta}$}.
        \end{cases}
    \end{equation}
    Then the error (as an estimate of $E_0$) of the expectation value of $H'$ with respect to $|\psi\rangle$ is upper bounded as
    \begin{equation}
    \label{approx_proj_bound_3}
        \frac{\langle\psi|(H'-E_0)|\psi\rangle}{\langle\psi|\psi\rangle}\le\delta'\mathbf{1}(\delta'>\widetilde{\Delta})+\|H'-H\|+6\|H\|\left(\frac{\|H'-H\|}{\delta'-\delta}+\frac{\zeta}{\||\psi\rangle\|^2}+\frac{8}{\||\psi\rangle\|^2}\left(1+\frac{\pi\delta}{R}\right)^{-2d}\right),
    \end{equation}
    where 
    \begin{equation}
    \label{zeta_def}
        \zeta\coloneqq2D\left(\epsilon+\|\textbf{S}\,'-\textbf{S}\,\|\right)
    \end{equation}
    and the bound holds for any parameters $0<\delta<\delta'<\|H\|$, provided
    \begin{equation}
    \label{thm3_assumption}
        \|H'-H\|<\delta'-\delta.
    \end{equation}
\end{theorem}
\noindent
The proof can be found in \cref{proofs_app}.

Our main result below is almost a corollary of \cref{approx_projector_perturbed_H_thm}.
It is obtained by inserting the particular effective Hamiltonian $H'$ constructed in \cref{effective_for_upper} into \cref{approx_projector_perturbed_H_thm}, and expressing the bound entirely in terms of problem parameters and noise rates.
\begin{theorem}
\label{complete_bound_thm}
    Let $H$ be a Hamiltonian, let $(\textbf{H},\textbf{S}\,)=(\textbf{V}\,^\dagger H\textbf{V}, \textbf{V}\,^\dagger\textbf{V}\,)$ be a real-time Krylov matrix pair representing $H$ in the Krylov space $\text{span}(\textbf{V}\,)$, and let $(\textbf{H}\,',\textbf{S}\,')$ be a Hermitian approximation to $(\textbf{H},\textbf{S}\,)$.
    Let $E_0$ be the ground state energy of $H$, which we want to estimate.
    Let $\epsilon>0$ be a regularization threshold, and let
    \begin{equation}
    \label{chi_def}
        \chi\coloneqq\|\textbf{H}\,'-\textbf{H}\,\|+\|\textbf{S}\,'-\textbf{S}\,\|\|H\|
    \end{equation}
    be a measure of the noise.
    Let
    \begin{equation}
    \label{perturbed_overlap}
        |\gamma_0'|^2\coloneqq|\gamma_0|^2-2\epsilon-2\|\textbf{S}\,'-\textbf{S}\,\|
    \end{equation}
    be a noisy effective version of the initial state's overlap $|\gamma_0|^2$ with the true ground state.
    Let
    \begin{equation}
    \label{perturbed_gap}
        \Delta'\coloneqq\Delta-\frac{\chi}{|\gamma_0'|^2}
    \end{equation}
    be a noisy effective version of the spectral gap $\Delta$ of $H$.
    Then the lowest eigenvalue $\widetilde{E}_0$ of the thresholded matrix pair obtained from $(\textbf{H}\,',\textbf{S}\,')$ is bounded as
    \begin{equation}
    \label{approx_proj_bound}
        \widetilde{E}_0-E_0\le\delta'\mathbf{1}(\delta'>\Delta')+\frac{\chi}{|\gamma_0'|^2}+\frac{6\|H\|}{|\gamma_0'|^2}\left(\frac{\chi}{\delta'-\delta}+\zeta+8\left(1+\frac{\pi\delta}{2\|H\|}\right)^{-2d}\right)
    \end{equation}
    where $\zeta$ is defined in \eqref{zeta_def} and the bound holds for any parameters $\delta'>\delta>0$, provided the following assumptions hold:\\
    (i)
    \begin{equation}
    \label{thm4_assumption}
        \frac{\chi}{|\gamma_0'|^2}<\delta'-\delta,
    \end{equation}
    (ii)
    \begin{equation}
    \label{thm4_assumption_2}
        \epsilon\ge\|\textbf{S}\,'-\textbf{S}\,\|,
    \end{equation}
    and (iii) the right-hand side of \eqref{perturbed_overlap} is positive.
\end{theorem}
\end{widetext}

\noindent
The proof is given in \cref{proofs_app}.
Note that the assumption \eqref{thm4_assumption} is extremely weak since, if it is violated, the error is of order $\Omega(\|H\|)$ due to the first term inside the square in \eqref{approx_proj_bound}.
The assumption \eqref{thm4_assumption_2} may be interpreted as formalizing the intuitive notion that the threshold should be larger than the noise rate, guaranteeing that we truncate out vectors that are compatible with zero under the noise.
Finally, the fact that we require the right-hand side of \eqref{perturbed_overlap} to be positive should not be surprising: if the initial state's overlap $|\gamma_0|^2$ with the true ground state were smaller than $\epsilon$ or $\|\textbf{S}'-\textbf{S}\|$, then it would be dominated by the error induced by the thresholding procedure or the noise, respectively.

The terms in \eqref{approx_proj_bound} possess intuitive origins:
\begin{enumerate}

    \item The first term $\delta'\mathbf{1}(\delta'>\Delta')$ is the size of the low energy subspace (above the ground state energy) that we project into, allowing for an extra tolerance (the difference between $\delta'$ and $\delta$) to account for the difference between the low energy eigenspaces of $H'$ and $H$.
    The indicator function $\mathbf{1}(\delta'>\Delta')$ captures the fact that if $\delta'\le\Delta'$, the effective gap, then this low energy subspace contains only the ground space, whose contribution to the energy error is captured by the second term.
    The theorem holds for any ${\delta'>\delta>0}$, with $\delta$ determining the rate of convergence (due to the last term inside the square).
    In words, the larger the low energy subspace, the faster we converge to it.

    \item The second term is an effect of the noise: it is due to the difference between the ground state energies of the exact Hamiltonian $H$ and the effective Hamiltonian corresponding to the noisy problem.
    
    \item The first two terms inside the large parentheses are also effects of the noise. Our ansatz state $|\psi\rangle$ in the effective Krylov space is an approximate ground state of the true Hamiltonian $H$, in the sense that it has high amplitude in low-energy eigenspaces (below $E_0+\delta$) of $H$, and low amplitude in high-energy eigenspaces. It is applied to the effective Hamiltonian $H'$. There are two distinct impacts on the resulting energy:
    \begin{enumerate}
    
        \item The fact that $|\psi\rangle$ is applied to $H'$ instead of $H$ means that its high amplitude in the low energy eigenspaces of $H$ is weakly mixed into the high energy eigenspaces (above $E_0+\delta'$) of $H'$.
        This leads to the first term inside the large parentheses: recall that $\chi$ as defined in \eqref{chi_def} determines the spectral norm distance between $H$ and $H'$ (see \eqref{Hprime_diff}).
        The gap $\delta'-\delta$ between ``low-energy'' and ``high-energy'' sets the rate of suppression.
        
        \item The second term ($\zeta$) inside the large parentheses comes from the disturbance of the (low) amplitudes of $|\psi\rangle$ in high-energy eigenspaces of $H'$, due to the error in the Krylov space and to the thresholding procedure.
        
    \end{enumerate}
    
    \item The final term is due to the ideal error of the quantum Krylov algorithm, from approximate projection of the initial state into the low-energy subspace.
    
\end{enumerate}

\cref{complete_bound_thm} still leaves us with a choice of the parameters $\delta$ and $\delta'$, which only pertain to the analysis, i.e., the result holds for any choices of their values and they are not required to actually execute the algorithm.
Hence if the problem parameters are known, one way to obtain an upper bound is to minimize over $\delta$ and $\delta'$ subjected to the constraints in the theorem statement.

Short of that, a reasonable choice would be
\begin{equation}
\label{project_to_gs_choice}
    \delta=\frac{\Delta'}{2},\quad\delta'=\Delta',
\end{equation}
which is the choice we would make if we want to obtain not just an approximate ground state energy but an approximate ground state.
In this case, the first term in \eqref{approx_proj_bound} vanishes, and substituting \eqref{project_to_gs_choice} into the remainder yields
\begin{equation}
\label{approx_proj_bound_gs}
\begin{split}
    &\widetilde{E}_0-E_0\\
    &~\le\frac{\chi}{|\gamma_0'|^2}+\frac{6\|H\|}{|\gamma_0'|^2}\left(\frac{2\chi}{\Delta'}+\zeta+8\left(1+\frac{\pi\Delta'}{4\|H\|}\right)^{-2d}\right).
\end{split}
\end{equation}

Eq.~\eqref{thm4_assumption_2} lower bounds the threshold; otherwise since $\zeta$ as defined in \eqref{zeta_def} is linear in $\epsilon$, we would choose $\epsilon\rightarrow0$.
Given \eqref{thm4_assumption_2}, the best we can do for the upper bound is choose $\epsilon=\|\textbf{S}'-\textbf{S}\|$, i.e. equality in \eqref{thm4_assumption_2}, and this represents a typical choice in practice as well, at least in scaling~\cite{klymko2022realtime,kirby2023exactefficient,epperly2021subspacediagonalization}.
With this choice, the bound \eqref{approx_proj_bound_gs} scales linearly with the noise rates $\|\textbf{H}'-\textbf{H}\|$ and $\|\textbf{S}'-\textbf{S}\|$.

To make this explicit, let us define a single unitless noise rate
\begin{equation}
    \eta\coloneqq\max\left(\|\textbf{S}'-\textbf{S}\|,\frac{\|\textbf{H}'-\textbf{H}\|}{\|H\|}\right),
\end{equation}
as in \cref{main_results}.
In terms of this, $\chi$ is bounded as $\chi\le O(\|H\|\eta)$ and $\zeta$ is bounded as ${\zeta\le O\left(D(\epsilon+\eta)\right)=O\left(D\eta\right)}$.
Inserting these into \eqref{approx_proj_bound_gs} and assuming $|\gamma_0'|^2=\Omega(|\gamma_0|^2)$ and $\Delta'=\Omega(\Delta)$ yields the asymptotic expression \eqref{asymptotic_bound} given in \cref{main_results}.

\section{Numerical example}
\label{numerics}

\begin{figure}
    \centering
    \includegraphics[width=\columnwidth]{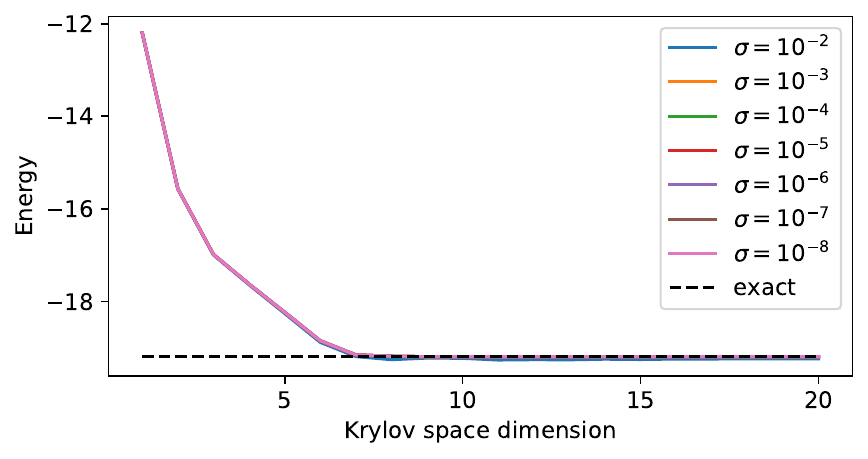}
    \includegraphics[width=\columnwidth]{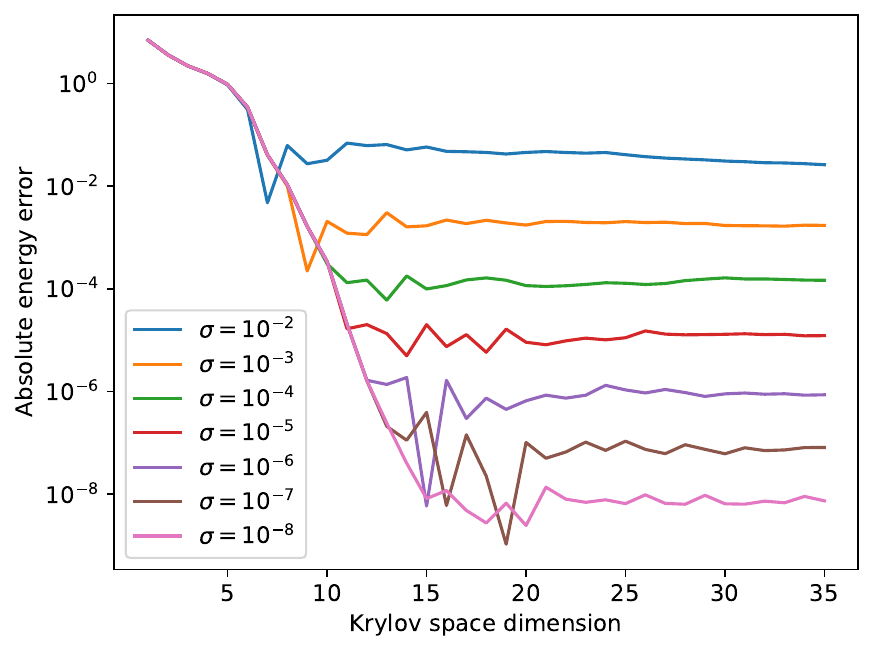}
    \includegraphics[width=\columnwidth]{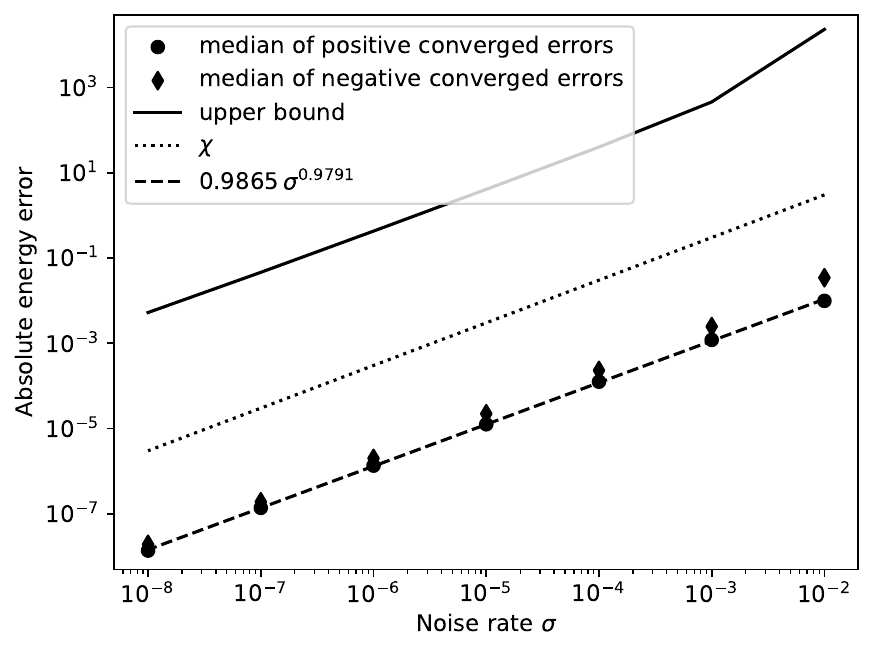}
    \caption{The top panel shows $\widetilde{E}_0$ versus $d$ for the Heisenberg model described in \cref{numerics}, classically simulated for several noise rates $\sigma$. Each point is a median of 10000 runs. The middle panel shows the corresponding absolute energy errors plotted on a log scale. The bottom panel shows converged energy errors, given by medians over all errors from dimensions 26 to 35, plotted against the noise rate. We separately evaluate these for the signed errors that are positive and negative. The dashed line is the best monomial fit to the positive error data. The solid curve shows the bound \eqref{approx_proj_bound}. Finally, the dotted curve shows the values of $\chi$ \eqref{chi_def} at each noise rate.}
    \label{example_fig}
\end{figure}

Although the main focus of this work is on the analytic bounds and their proofs, a numerical example is illustrative of both the application and the limitations of the results.
Source code for the following is available.\footnote{\url{https://github.com/wmkirby1/krylov-analysis-paper-numerics}}
We take as our example a Heisenberg model with spin anisotropy $j=1$ and a weak field strength of $h=0.2$:
\begin{equation}
    H=h\sum_mZ_m+\sum_{\langle m,n\rangle}\big(X_mX_n+Y_mY_n+jZ_mZ_n\big).
\end{equation}

We classically simulate the quantum Krylov algorithm for this model on a $3\times3$ square lattice.
For an initial state we take the antiferromagnetic state containing 4 spin-up ($|0\rangle$) and 5 spin-down ($|1\rangle$) sites, since the field gives this spin sector a lower energy than the opposite antiferromagnetic state (5 and 4).
This yields an overlap $|\gamma_0|^2\approx0.275$.
The relevant spectral gap is the gap between lowest and next-to-lowest energies in this sector, which is $\Delta\approx3.96$.

To assess the effects of errors in the matrix elements, we add Gaussian noise of various widths $\sigma$ to the matrix elements in $\textbf{S}$, and widths $\|H\|\sigma$ to the matrix elements in $\textbf{H}$.
The regularization threshold $\epsilon$ is chosen to be $0.1\,D\,\sigma$, which is an instantiation of $\epsilon=O(\|\textbf{S}'-\textbf{S}\|)$ that is effective in practice.
The remaining parameters in \eqref{approx_proj_bound} can be calculated from the above quantities.
Finally, the bound \eqref{approx_proj_bound} holds for any choices of $\delta$ and $\delta'$ subject to the constraints given in \cref{complete_bound_thm}, so we can find the tightest bound by minimizing over their values subject to those constraints.

The lower panel in \cref{example_fig} shows the converged errors, represented as the medians of all errors from dimensions 26 to 35 at each noise rate.
We separate these data into the positive signed errors and the negative signed errors, since we have different bounds for these two cases.
The plot also shows the best monomial fit to the positive error data, the bound obtained by optimizing \eqref{approx_proj_bound} over $\delta$ and $\delta'$, and the values of $\chi$ for each noise rate.
The best monomial fit to the data is $O(\sigma^{0.979})=O(\chi^{0.979})$, illustrating that the converged energy errors perform essentially as the expected $O(\chi)$ of the upper bound with the choice $\epsilon=O(\chi)$.
The bound exhibits nearly identical scaling, but is about six orders of magnitude worse than the actual errors.
As it turns out, $\chi$ alone also provides an upper bound in this case, but is only worse than the actual errors by about two orders of magnitude.

Some takeaways from this numerical demonstration are as follows.
First, the positive error data exhibits the same (nearly) linear scaling with noise as the bound, but the bound overshoots the actual errors by a constant on the order of $10^6$.
This is not too surprising because the present example is a specific instance and likely not a worst case, and also because the proof of the bound involves a sequence of intermediate inequalities.
Tracking any or all of these explicitly is possible in principle and would lead to a more complicated but tighter bound.

The other takeaway is that the negative error data exhibit nearly identical performance to the positive error data, even though the threshold is chosen as $\epsilon=O(\chi)$.
This illustrates a point discussed in \cref{main_results}: although the lower bound in \cref{lower_bound_thm} suggests that $\epsilon=O(\chi)$ could lead to negative errors of order $\|H\|$, in practice we typically do not see this.
For this reason, we do not plot the lower bound's magnitude in \cref{example_fig} because it is roughly 148, independent of $\sigma$.

This emphasizes that the lower bound in this work can likely be improved, as discussed in \cref{main_results}.
The values of $\chi$, which are approximately equivalent to the numerator of the lower bound \eqref{approx_lower_bound}, appear to provide a bound in this instance: this is merely suggestive since it is a single example case, but nonetheless one may hope that the lower bound can be tightened to some function does not scale as $O(1/\epsilon)$.
The bound would still need to account for the fact that in practice we often do see large negative fluctuations if $\epsilon$ is made too small.
Hence one should not hope to eliminate the dependence on $\epsilon$ from the lower bound entirely, but it could have some alternative dependence that accounts for the observed performance.

\section{Conclusion}

Although this work focused on real-time Krylov spaces, since they are the most feasible version of quantum Krylov for noisy quantum computers due to the possibility of low circuit depths, extending to other types of Krylov spaces would be straightforward.
The results in \cref{noise_as_perturbation} on effective Krylov spaces and Hamiltonians only require that the noisy matrix pair is Hermitian (which an appropriate construction can guarantee), and are agnostic to the underlying ideal Krylov space.
For the error bounds in \cref{upper_bound_sec}, the same is true with the exception of \cref{epperly_thm}, which shows existence of an approximate ground state in the effective Krylov space.
Since the following theorems assume access to an approximate ground state with the specific properties of the one given in \cref{epperly_thm}, they would also potentially need to be modified.
However, at least for Krylov spaces spanned by powers of the Hamiltonian, the construction of low-energy states in~\cite{saad1980lanczos} is similar to that of \cref{upper_bound_sec}.
It would be an interesting exercise to modify \cref{upper_bound_sec} to use the construction for polynomials rather than complex exponentials and check whether the results substantively differ.

As discussed in \cref{main_results,numerics}, another direction for future work is to tighten both bounds, but particularly the lower bound.
The lower bound should be the focus because with the typical in-practice choice of threshold proportional to error rate, the lower bound becomes trivial.
On the other hand, at least in some cases of interest, the actual performance of the method does not suffer with this choice, as illustrated in \cref{numerics} as well as prior work~\cite{klymko2022realtime,kirby2023exactefficient,epperly2021subspacediagonalization}.
In contrast, the upper bound's scaling with error rate and threshold might now be optimal, but it is clearly loose in constant factors, which one might hope to improve.

As for impact on users of quantum Krylov algorithms, this work can help to clarify what one should expect from the performance of these methods, at least in terms of scaling.
The practical takeaway of \cref{upper_bound_sec} is that if error rates are small enough for the lowest energy in the Krylov space to be resolved (i.e., for the conditions of \cref{complete_bound_thm} to hold), then one should expect to see an energy versus Krylov dimension curve qualitatively similar to the top panel in \cref{example_fig}.
In practice those conditions may be difficult to evaluate, and it may also be difficult to know \emph{a priori} that the threshold is large enough to avoid negative fluctuations due to ill-conditioning.
However, the existence of conditions on the noise and threshold that guarantee exponential decay with Krylov dimension is still useful because one can then look for the exponential decay as a signature of a successful run.

Prior works that compare the quantum Krylov algorithm with other quantum algorithms for ground state estimation either do so numerically or assume that the true scaling of energy error with respect to input error rate is linear~\cite{klymko2022realtime,kirby2023exactefficient,motta2023subspace}.
The present work's contribution to this dialogue is to confirm the latter, at least for the upper bound and the usual choice of threshold scaling.

Otherwise, the main points regarding comparison to other ground state estimation algorithms remain the same as in prior work, so here we will merely review the highlights.
The primary advantages of the quantum Krylov algorithm are its potential for low-depth circuits using Trotterized time-evolutions, and its noise robustness both in theory and in practice.
Its primary disadvantage with respect to fault-tolerant quantum algorithms like quantum phase estimation and other techniques that achieve the Heisenberg limit (e.g.,~\cite{dong2022groundstate,lin2022heisenberglimited,li2023esprit}) is that the quantum Krylov algorithm uses repeated sampling to estimate the matrix elements.
Hence its error will scale as $O(1/\sqrt{T})$ for total runtime $T$, as opposed the Heisenberg limit of $O(1/T)$ achieved by quantum phase estimation.
Finally, the dependence of \eqref{approx_proj_bound} on the initial state's overlap with the true ground state provides the limitation that prevents the algorithm from efficiently solving QMA-complete problems.
This is in common with nearly all other quantum algorithms for ground state estimation, with the exception of adiabatic state preparation~\cite{farhi2000adiabatic}.

To sum up, in this work we provided a new error analysis for the real-time quantum Krylov algorithm in the presence of noise, using eigenvalue thresholding.
The main advance over prior results~\cite{epperly2021subspacediagonalization} is obtaining linear scaling of the upper bound on the signed energy error, with respect to the noise rate.
This brings the theoretical analysis closer to alignment with the numerics of prior works~\cite{epperly2021subspacediagonalization,klymko2022realtime,kirby2023exactefficient}.
In addition, the technique of expressing error in a Krylov matrix pair in terms an effective Hamiltonian and an effective Krylov space may be more broadly useful.

\begin{acknowledgements}
I am grateful to Nobuyuki Yoshioka, Mario Motta, Antonio Mezzacapo, Kunal Sharma, Minh Tran, Patrick Rall, and Ethan Epperly for helpful conversations.
I owe Ethan a particularly important thank-you for pointing out an error in the first version of this paper, and assisting in correcting it.
I also especially thank Patrick for proofreading not one, but two revisions.
Finally, I thank the anonymous reviewers, who provided useful feedback.
\end{acknowledgements}


\bibliographystyle{quantum}
\bibliography{references}

\appendix

\begin{widetext}

\section{Proofs}
\label[appendix]{proofs_app}

We begin by stating two classic results of matrix analysis, for convenience:
\noindent
\begin{lemma}[version of Weyl's Theorem, originally in~\cite{weyl1912asymptotische}, see also~\cite{bhatia1997matrix}, Cor. III.2.6]
\label{weyl_thm}
    Let $H$ and $H'$ be Hermitian matrices of the same dimensions, and let $E_i,E_i'$ be their eigenvalues in weakly increasing order. Then for any $i$,
    \begin{equation}
        |E_i'-E_i|\le\|H'-H\|.
    \end{equation}
\end{lemma}

\noindent
\begin{lemma}[special case of Davis-Kahan ``$\sin\Theta$ theorem,'' originally in~\cite{davis1970rotation}, see also~\cite{bhatia1997matrix}, Thm. VII.3.1]
\label{davis_kahan}
    Let $H$ and $H'$ be Hermitian matrices of the same dimensions, and let $\Pi_K,\Pi_{K'}'$ be their spectral projectors onto subsets $K$ and $K'$ of the real line that are separated by a gap $\delta>0$, i.e., there exists $a\in\mathds{R}$ such that (without loss of generality) $k\le a$ and $k'\ge a+\delta$ for all $k\in K,k'\in K'$. Then
    \begin{equation}
        \|\Pi_K\Pi_{K'}'\|\le\frac{\|H'-H\|}{\delta}.
    \end{equation}
\end{lemma}

\noindent
We now proceed to proofs of the results in the main text.\\~\\

\noindent
\textbf{Theorem~\ref{Hprime_diff_thm_lower}}.
\emph{
    Let the unitary $G$ in the definition \eqref{Vprime_def_lower} of $\textbf{V}\,'$ be defined such that
    \begin{equation}
    \label{G_def_app}
        \sqrt{\textbf{S}}\,\Pi'=G\sqrt{\Pi'\textbf{S}\,\Pi'}
    \end{equation}
    is the polar decomposition of $\sqrt{\textbf{S}}\,\Pi'$.
    Assume that $\|\textbf{S}\,'-\textbf{S}\,\|\le\epsilon$.
    Then for $H'$ as defined in \eqref{Hprime_def},
    \begin{equation}
        \|H'-H\|\le\frac{\|\textbf{H}\,'-\textbf{H}\,\|+(1+\sqrt{2})\|\textbf{S}\,'-\textbf{S}\,\|\|H\|}{\epsilon}.
    \end{equation}
}
\begin{proof}

First, note that \eqref{G_def_app} is a valid polar decomposition of $\sqrt{\textbf{S}}\Pi'$ because
\begin{equation}
\label{polar_decomp_remark_eq}
    (\sqrt{\textbf{S}}\Pi')^\dagger(\sqrt{\textbf{S}}\Pi')=\Pi'\textbf{S}\Pi'.
\end{equation}
The singular value decomposition of $\sqrt{\textbf{S}}\Pi'$ is
\begin{equation}
    \sqrt{\textbf{S}}\Pi'=UDV^\dagger
\end{equation}
for some unitaries $U,V$ and diagonal, nonnegative $D$.
The polar decomposition can be constructed from this as
\begin{equation}
\label{polar_decomp_remark_eq_2}
    \sqrt{\textbf{S}}\Pi'=\underbrace{(UV^\dagger)}_{G}(VDV^\dagger),
\end{equation}
and hence
\begin{equation}
    (\sqrt{\textbf{S}}\Pi')^\dagger(\sqrt{\textbf{S}}\Pi')=(VDV^\dagger)^2.
\end{equation}
Thus since $VDV^\dagger$ is p.s.d., by \eqref{polar_decomp_remark_eq}
\begin{equation}
    VDV^\dagger=\sqrt{\Pi'\textbf{S}\Pi'}.
\end{equation}
Inserting this into \eqref{polar_decomp_remark_eq_2} yields \eqref{G_def_app}.

Proceeding to the main proof, for convenience we repeat the definition \eqref{Hprime_def} of $H'$:
\begin{equation}
\label{Hprime_def_app}
    H'=H+\textbf{V}'\textbf{S}''^+\left(\textbf{H}'-\textbf{V}'^\dagger H\textbf{V}'\right)\textbf{S}''^+\textbf{V}'^\dagger.
\end{equation}
Subtracting $H$ from both sides of \eqref{Hprime_def_app}, we have
\begin{equation}
\label{H_prime_diff_1}
    \|H'-H\|=\|\textbf{V}'\textbf{S}''^+\left(\textbf{H}'-\textbf{V}'^\dagger H\textbf{V}'\right)\textbf{S}''^+\textbf{V}'^\dagger\|.
\end{equation}
Inserting the definition \eqref{Vprime_def_lower} of $\textbf{V}'\coloneqq FG\sqrt{\textbf{S}''}$ yields
\begin{equation}
\label{H_prime_diff_2}
\begin{split}
    \|H'-H\|&=\|FG\sqrt{\textbf{S}''}\textbf{S}''^+\left(\textbf{H}'-\textbf{V}'^\dagger H\textbf{V}'\right)\textbf{S}''^+\sqrt{\textbf{S}''}G^\dagger F^\dagger\|\\
    &=\|\sqrt{\textbf{S}''^+}\left(\textbf{H}'-\textbf{V}'^\dagger H\textbf{V}'\right)\sqrt{\textbf{S}''^+}\|,
\end{split}
\end{equation}
where the second step follows because $F$ has orthonormal columns and $G$ is unitary.
Next, using ${\textbf{H}\coloneqq\textbf{V}^\dagger H\textbf{V}}$,
\begin{equation}
\label{H_prime_diff_3_lower}
\begin{split}
    \|H'-H\|&\le\|\sqrt{\textbf{S}''^+}\left(\textbf{H}'-\textbf{H}\right)\sqrt{\textbf{S}''^+}\|+\|\sqrt{\textbf{S}''^+}\left(\textbf{V}^\dagger H\textbf{V}-\textbf{V}'^\dagger H\textbf{V}'\right)\sqrt{\textbf{S}''^+}\|\\
    &\le\|\sqrt{\textbf{S}''^+}\|\|\textbf{H}'-\textbf{H}\|\|\sqrt{\textbf{S}''^+}\|+\|\sqrt{\textbf{S}''^+}\textbf{V}^\dagger H\left(\textbf{V}-\textbf{V}'\right)\sqrt{\textbf{S}''^+}\|+\|\sqrt{\textbf{S}''^+}\left(\textbf{V}^\dagger-\textbf{V}'^\dagger\right)H\textbf{V}'\sqrt{\textbf{S}''^+}\|.
\end{split}
\end{equation}

We upper bound the three terms in \eqref{H_prime_diff_3_lower} separately.
The first term in \eqref{H_prime_diff_3_lower} is upper bounded as
\begin{equation}
\label{first_term}
    \|\sqrt{\textbf{S}''^+}\|\|\textbf{H}'-\textbf{H}\|\|\sqrt{\textbf{S}''^+}\|\le\frac{\|\textbf{H}'-\textbf{H}\|}{\epsilon},
\end{equation}
since the smallest nonzero eigenvalue of $\textbf{S}''$ is at least $\epsilon$, by construction \eqref{Sprimeprime_def}.

The second term in \eqref{H_prime_diff_3_lower} is upper bounded as
\begin{equation}
\label{second_term}
\begin{split}
    \|\sqrt{\textbf{S}''^+}\textbf{V}^\dagger H\left(\textbf{V}-\textbf{V}'\right)\sqrt{\textbf{S}''^+}\|&=\|\sqrt{\textbf{S}''^+}\textbf{V}^\dagger H\left(\textbf{V}-\textbf{V}'\right)\Pi'\sqrt{\textbf{S}''^+}\|\\
    &\le\|\sqrt{\textbf{S}''^+}\textbf{V}^\dagger\|\|H\|\|\textbf{V}\Pi'-\textbf{V}'\Pi'\|\|\sqrt{\textbf{S}''^+}\|\\
    &\le\frac{\|\sqrt{\textbf{S}''^+}\textbf{V}^\dagger\|\|H\|\|\textbf{V}\Pi'-\textbf{V}'\Pi'\|}{\sqrt{\epsilon}},
\end{split}
\end{equation}
where the first step follows because $\sqrt{\textbf{S}''^+}$ has the same range as $\textbf{S}''$, and hence $\Pi'\sqrt{\textbf{S}''^+}=\sqrt{\textbf{S}''^+}$, and the last step follows because the smallest nonzero eigenvalue of $\textbf{S}''$ is at least $\epsilon$.
Continuing, we insert the polar decompositions of $\textbf{V}$ and $\textbf{V}'$ as in \eqref{V_polar} and \eqref{Vprime_def_lower}, to obtain
\begin{equation}
\label{second_term_2}
\begin{split}
    \|\sqrt{\textbf{S}''^+}\textbf{V}^\dagger H\left(\textbf{V}-\textbf{V}'\right)\sqrt{\textbf{S}''^+}\|&\le\frac{\|\sqrt{\textbf{S}''^+}\sqrt{\textbf{S}}F^\dagger\|\|H\|\|F\sqrt{\textbf{S}}\Pi'-FG\sqrt{\textbf{S}''}\Pi'\|}{\sqrt{\epsilon}}\\
    &=\frac{\|\sqrt{\textbf{S}''^+}\sqrt{\textbf{S}}\|\|H\|\|G^\dagger\sqrt{\textbf{S}}\Pi'-\sqrt{\textbf{S}''}\Pi'\|}{\sqrt{\epsilon}},
\end{split}
\end{equation}
where the second step follows because $F$ has orthonormal columns and $G$ is unitary.

We now bound the factors in the numerator of \eqref{second_term_2} separately.
For ${\|\sqrt{\textbf{S}''^+}\sqrt{\textbf{S}}\|}$, using the fact ${\|A\|=\sqrt{\|AA^\dagger\|}}$ for the spectral norm and any matrix $A$,
\begin{equation}
\label{approx_inv}
\begin{split}
    \|\sqrt{\textbf{S}''^+}\sqrt{\textbf{S}}\|&=\sqrt{\|\sqrt{\textbf{S}''^+}\textbf{S}\sqrt{\textbf{S}''^+}\|}\\
    &\le\sqrt{\|\sqrt{\textbf{S}''^+}\textbf{S}''\sqrt{\textbf{S}''^+}\|+\|\sqrt{\textbf{S}''^+}(\textbf{S}-\textbf{S}'')\sqrt{\textbf{S}''^+}\|}\\
    &=\sqrt{1+\|\sqrt{\textbf{S}''^+}(\textbf{S}-\textbf{S}')\sqrt{\textbf{S}''^+}\|}\\
    &\le\sqrt{1+\|\sqrt{\textbf{S}''^+}\|\|\textbf{S}'-\textbf{S}\|\|\sqrt{\textbf{S}''^+}\|}\\
    &\le\sqrt{1+\frac{\|\textbf{S}'-\textbf{S}\|}{\epsilon}}\\
    &\le\sqrt{2},
\end{split}
\end{equation}
where the third line follows because $\textbf{S}''=\Pi'\textbf{S}'\Pi'$ and $\Pi'\sqrt{\textbf{S}''^+}=\sqrt{\textbf{S}''^+}\Pi'=\sqrt{\textbf{S}''^+}$ (as discussed above),
the fifth line follows because the smallest nonzero eigenvalue of $\textbf{S}''$ is $\epsilon$, and the final step follows by the assumption in the theorem statement.

For $\|G^\dagger\sqrt{\textbf{S}}\Pi'-\sqrt{\textbf{S}''}\Pi'\|$,
\begin{equation}
\label{split_again}
    \|G^\dagger\sqrt{\textbf{S}}\Pi'-\sqrt{\textbf{S}''}\Pi'\|=\|\sqrt{\Pi'\textbf{S}\Pi'}-\sqrt{\Pi'\textbf{S}''\Pi'}\|,
\end{equation}
by \eqref{G_def_app} and the fact that $\Pi'\textbf{S}''=\textbf{S}''\Pi'=\textbf{S}''$.
We can upper bound this using an inequality of van Hemmen and Ando~\cite[Proposition 3.2]{van_Hemmen1980inequality}, applied only to the submatrices of $\sqrt{\Pi'\textbf{S}\Pi'}$ and $\sqrt{\Pi'\textbf{S}''\Pi'}$ within the range of $\Pi'$ (outside that range $\sqrt{\Pi'\textbf{S}\Pi'}$ and $\sqrt{\Pi'\textbf{S}''\Pi'}$ are zero and thus equal).
With $0$ and $\sqrt{\epsilon}$ being lower bounds on the least eigenvalues of these submatrices, respectively (the former because $\textbf{S}$ is p.s.d.), the inequality~\cite[Proposition 3.2]{van_Hemmen1980inequality} yields
\begin{equation}
\label{sqrt_diff}
\begin{split}
    \|G^\dagger\sqrt{\textbf{S}}\Pi'-\sqrt{\textbf{S}''}\Pi'\|&\le\frac{\|\Pi'\textbf{S}\Pi'-\Pi'\textbf{S}''\Pi'\|}{\sqrt{\epsilon}}\\
    &=\frac{\|\Pi'\textbf{S}\Pi'-\Pi'\textbf{S}'\Pi'\|}{\sqrt{\epsilon}}\\
    &\le\frac{\|\textbf{S}-\textbf{S}'\|}{\sqrt{\epsilon}},
\end{split}
\end{equation}
where the second line follows by \eqref{Sprimeprime_def}.
Inserting \eqref{approx_inv} and \eqref{sqrt_diff} into \eqref{second_term_2} yields the following upper bound on the second term in \eqref{H_prime_diff_3_lower}:
\begin{equation}
\label{second_term_3}
\begin{split}
    \|\sqrt{\textbf{S}''^+}\textbf{V}^\dagger H\left(\textbf{V}-\textbf{V}'\right)\sqrt{\textbf{S}''^+}\|&\le\frac{\sqrt{2}\|H\|\|\textbf{S}-\textbf{S}'\|}{\epsilon}.
\end{split}
\end{equation}

For the third term in \eqref{H_prime_diff_3_lower}, we follow the same derivation as in \eqref{second_term} and \eqref{second_term_2}, just for the adjoints, and obtain
\begin{equation}
\label{third_term}
    \|\sqrt{\textbf{S}''^+}\left(\textbf{V}^\dagger-\textbf{V}'^\dagger\right)H\textbf{V}'\sqrt{\textbf{S}''^+}\|\le\frac{\|\Pi'\sqrt{\textbf{S}}G-\Pi'\sqrt{\textbf{S}''}\|\|H\|\|\sqrt{\textbf{S}''}\sqrt{\textbf{S}''^+}\|}{\sqrt{\epsilon}}.
\end{equation}
This is simpler than \eqref{second_term_2} because $\|\sqrt{\textbf{S}''}\sqrt{\textbf{S}''^+}\|=1$ immediately, and additionally inserting \eqref{sqrt_diff} for the first factor in the numerator (which is the adjoint of the left-hand side in \eqref{sqrt_diff}) yields
\begin{equation}
\label{third_term_2}
    \|\sqrt{\textbf{S}''^+}\left(\textbf{V}^\dagger-\textbf{V}'^\dagger\right)H\textbf{V}'\sqrt{\textbf{S}''^+}\|\le\frac{\|H\|\|\textbf{S}-\textbf{S}'\|}{\epsilon}.
\end{equation}
Inserting the bounds \eqref{first_term}, \eqref{second_term_3}, and \eqref{third_term_2} for all three terms into \eqref{H_prime_diff_3_lower} yields our final bound of
\begin{equation}
    \|H'-H\|\le\frac{\|\textbf{H}'-\textbf{H}\|+(1+\sqrt{2})\|\textbf{S}'-\textbf{S}\|\|H\|}{\epsilon}.
\end{equation}

\end{proof}

\noindent
\textbf{Theorem~\ref{epperly_thm}}~[partly derived from Theorem 3.1 in~\cite{epperly2021subspacediagonalization}].
\emph{
    Let $d$ be a positive integer defining the dimension $D=2d+1$ as above, let $\delta>0$, let $(E_k,|E_k\rangle)$ be the eigenpairs of $H$ in weakly increasing order of energy, and let $R\coloneqq E_\text{max}-E_0$ be the spectral range of $H$.
    Let
    \begin{equation}
        |\psi_0\rangle=\sum_{k=0}^{N-1}\gamma_k|E_k\rangle
    \end{equation}
    be the expansion of $|\psi_0\rangle$ in the energy eigenbasis of $H$, where $N$ is the Hilbert space dimension.
    Assume
    \begin{equation}
    \label{thm1_assumption_app}
        \|\textbf{S}\,'-\textbf{S}\,\|\le\epsilon.
    \end{equation}
    Then there exists an operator $P$ such that the column space of $\textbf{V}\,'$ contains a state
    \begin{equation}
    \label{ansatz_approx_app}
        |\psi\rangle=P|\psi_0\rangle
    \end{equation}
    and $P$ satisfies
    \begin{equation}
    \label{agsp_def_part_1_app}
        P|E_k\rangle=\beta'_k|E_k\rangle,
    \end{equation}
    where
    \begin{equation}
    \label{agsp_def_part_2_app}
        |\beta'_k|^2\le
        \begin{cases}
            2+\alpha_k\quad\text{if $E_k-E_0<\delta$},\\
            8\left(1+\frac{\pi\delta}{R}\right)^{-2d}+\alpha_k\quad\text{if $E_k-E_0\ge\delta$}.
        \end{cases}
    \end{equation}
    The $\alpha_k$ satisfy
    \begin{equation}
    \label{agsp_def_part_3_app}
        \sum_{k=0}^{N-1}|\gamma_k|^2\alpha_k\le2D\left(\epsilon+\|\textbf{S}\,'-\textbf{S}\,\|\right).
    \end{equation}
    The norm of $|\psi\rangle$ is can be lower bounded with or without explicit dependence on $c'$, the coordinates of $|\psi\rangle$ in the column space of $\textbf{V}\,'$:
    \begin{equation}
    \label{psi_norm_lower_bound_app}
    \begin{split}
        &\||\psi\rangle\|^2\ge\|c'\|^2\left(|\gamma_0|^2-\epsilon-\|\textbf{S}\,'-\textbf{S}\,\|\right),\\
        &\||\psi\rangle\|^2\ge|\gamma_0|^2-2\epsilon-2\|\textbf{S}\,'-\textbf{S}\,\|.
    \end{split}
    \end{equation}
}
\begin{proof}

By Lemma 3.3 in~\cite{epperly2021subspacediagonalization}, for any positive integer $d$ and parameter $0<a<\pi$ there exists a degree-$d$ trigonometric polynomial $p^*$ whose magnitude is everywhere upper bounded by $1$, satisfying
\begin{equation}
    p^*(0)=1
\end{equation}
and
\begin{equation}
    |p^*(\theta)|\le2(1+a)^{-d}\quad\text{for all $\theta\in(-\pi,\pi],~|\theta|\ge a$}.
\end{equation}
Let
\begin{equation}
\label{fourier_series}
    p^*\left(\frac{\pi(E-E_0)}{R}\right)=p^*((E-E_0)dt)=\sum_{j=-d}^dc_je^{ijE\,dt}
\end{equation}
be the Fourier transform of $p^*((E-E_0)dt)$, where ${dt\coloneqq\pi/R}$ is defined by the spectral range $R$, so that the full argument ${(E-E_0)dt\in[0,\pi]}$ for all ${E\in[E_0,E_\text{max}]}$.
Choose $a=\delta\,dt=\frac{\pi\delta}{R}$.
By the definition of $p^*$, this implies that
\begin{equation}
\label{preserving}
    \sum_{j=-d}^dc_je^{ijE_0\,dt}=1
\end{equation}
and
\begin{equation}
\label{shrinking}
    \left|\sum_{j=-d}^dc_je^{ijE\,dt}\right|\le2\left(1+\frac{\pi\delta}{R}\right)^{-d}\quad\text{for all}~E\ge E_0+\delta.
\end{equation}
In the ideal Krylov space, our ansatz would be $\textbf{V}c$, which we could show based on the above definitions to be an approximate ground state projector for $H$~\cite{epperly2021subspacediagonalization}.

Instead, we consider a modified set of coordinates in the effective Krylov space: $c'$, defined by
\begin{equation}
\label{cprime_def}
    c'\coloneqq\Pi'c\sqrt{\frac{c'^\dagger\textbf{S}c'}{c'^\dagger\textbf{S}'c'}}=\widetilde{c}\sqrt{\frac{\widetilde{c}\,^\dagger\textbf{S}\widetilde{c}}{\widetilde{c}\,^\dagger\textbf{S}'\widetilde{c}}}\quad\text{for}\quad\widetilde{c}\coloneqq\Pi'c,
\end{equation}
where $\Pi'$ is the projector onto the eigenspaces of $\textbf{S}'$ with eigenvalues above threshold, i.e., onto the range of $\textbf{S}''$.
Note that this choice and its consequences are the main difference from the proof of Theorem 3.1 in~\cite{epperly2021subspacediagonalization}, since in that proof the same coordinates are used in the perturbed Krylov space as in the ideal Krylov space.

In terms of this $c'$, we define the effective Krylov space as in \eqref{Vprime_def}, which we repeat here for convenience:
\begin{equation}
\label{Vprime_def_app}
    \textbf{V}'\coloneqq\sqrt{\frac{c'^\dagger\textbf{S}'c'}{c'^\dagger\textbf{S}c'}}\,\textbf{V}=\sqrt{\frac{c'^\dagger\textbf{S}''c'}{c'^\dagger\textbf{S}c'}}\,\textbf{V},
\end{equation}
where the second equality follows because $\Pi'\textbf{S}''\Pi'=\Pi'\textbf{S}'\Pi'$ by definition \eqref{Sprimeprime_def}.
The ansatz vector is then
\begin{equation}
\label{psi_def_2}
    |\psi\rangle\coloneqq\textbf{V}'c'=\textbf{V}\widetilde{c},
\end{equation}
where the second equality follows from \eqref{cprime_def} and \eqref{Vprime_def_app}.

First we want to lower bound the norm of $|\psi\rangle$: by \eqref{cprime_def} and \eqref{psi_def_2},
\begin{equation}
\label{psi_norm_lower_bound_num_1}
\begin{split}
    \frac{\langle\psi|\psi\rangle}{\|c'\|^2}&=\frac{\widetilde{c}\,^\dagger\textbf{S}\widetilde{c}}{\|c'\|^2}
    =\frac{\widetilde{c}\,^\dagger\textbf{S}'\widetilde{c}}{\|\widetilde{c}\|^2}\\
    &\ge\frac{\widetilde{c}\,^\dagger\textbf{S}'\widetilde{c}}{\|c\|^2}\\
    &\ge\widetilde{c}\,^\dagger\textbf{S}'\widetilde{c}
    =c^\dagger\textbf{S}''c
    =c^\dagger\textbf{S}c+c^\dagger(\textbf{S}''-\textbf{S})c\\
    &\ge c^\dagger\textbf{S}c-\|c\|^2\|\textbf{S}''-\textbf{S}\|\\
    &\ge c^\dagger\textbf{S}c-\|\textbf{S}''-\textbf{S}\|\\
    &\ge c^\dagger\textbf{S}c-\|\textbf{S}''-\textbf{S}'\|-\|\textbf{S}'-\textbf{S}\|\\
    &\ge c^\dagger\textbf{S}c-\epsilon-\|\textbf{S}'-\textbf{S}\|,
\end{split}
\end{equation}
which twice uses $\|c\|^2\le1$ (as argued in~\cite{epperly2021subspacediagonalization}, by Parseval's Theorem and Proposition 3.4 in~\cite{epperly2021subspacediagonalization}).
The final step follows because $\textbf{S}''=\Pi'\textbf{S}'\Pi'$ by definition \eqref{Sprimeprime_def}, and thus $\textbf{S}''-\textbf{S}'$ is supported only on the nullspace of $\textbf{S}''$; the eigenvalues of $\textbf{S}'$ in this subspace lie between $-\|\textbf{S}'-\textbf{S}\|\ge-\epsilon$ and $\epsilon$, with the inequality following by Weyl's theorem (\cref{weyl_thm}) and the fact that $\textbf{S}$ is p.s.d., as well as \eqref{thm1_assumption_app}.

To obtain a lower bound without explicit dependence on $\|c'\|$, we instead use
\begin{equation}
\label{psi_norm_lower_bound_num_2}
\begin{split}
    \langle\psi|\psi\rangle&=\widetilde{c}^\dagger\textbf{S}\widetilde{c}\\
    &=c^\dagger\Pi'\textbf{S}\Pi'c\\
    &=c^\dagger\textbf{S}c+c^\dagger\left(\Pi'\textbf{S}''\Pi'-\textbf{S}\right)c+c^\dagger\left(\Pi'\textbf{S}\Pi'-\Pi'\textbf{S}''\Pi'\right)c\\
    &=c^\dagger\textbf{S}c+c^\dagger\left(\textbf{S}''-\textbf{S}\right)c-c^\dagger\Pi'\left(\textbf{S}''-\textbf{S}\right)\Pi'c\\
    &\ge c^\dagger\textbf{S}c-2\|c\|^2\|\textbf{S}''-\textbf{S}\|\\
    &\ge c^\dagger\textbf{S}c-2\epsilon-2\|\textbf{S}'-\textbf{S}\|,
\end{split}
\end{equation}
using the same upper bound $\|\textbf{S}''-\textbf{S}\|\le\epsilon+\|\textbf{S}'-\textbf{S}\|$ as in \eqref{psi_norm_lower_bound_num_1}.
Finally,
\begin{equation}
    c^\dagger\textbf{S}c=\|\textbf{V}c\|^2=\left\|\sum_{k=0}^{N-1}\gamma_k\sum_{j=-d}^dc_je^{ijE_k\,dt}|E_k\rangle\right\|^2=\sum_{k=0}^{N-1}\left|\gamma_k\sum_{j=-d}^dc_je^{ijE_k\,dt}\right|^2\ge|\gamma_0|^2,
\end{equation}
where the last step follows by \eqref{preserving}. 
Inserting this into \eqref{psi_norm_lower_bound_num_1} and \eqref{psi_norm_lower_bound_num_2} yields the first and second lines of \eqref{psi_norm_lower_bound_app}, respectively.

Next we want to upper bound the coefficients of $|\psi\rangle$ in the energy eigenbasis.
Let
\begin{equation}
\label{psi_0_decomp}
    |\psi_0\rangle=\sum_{k=0}^{N-1}\gamma_k|E_k\rangle
\end{equation}
be the expansion of $|\psi_0\rangle$ in the energy eigenbasis of $H$, and let
\begin{equation}
    |\psi\rangle=\sum_{k=0}^{N-1}\beta'_k\gamma_k|E_k\rangle
\end{equation}
be the expansion of $|\psi\rangle$ in the energy eigenbasis of $H$.
We will be aiming to upper bound the magnitudes of the $\beta'_k$.

By the definition \eqref{exact_krylov_space} of $\textbf{V}$, and using the second form in \eqref{psi_def_2} for $|\psi\rangle$,
\begin{equation}
    |\psi\rangle=\textbf{V}\widetilde{c}=\sum_{j=-d}^d\widetilde{c}_je^{ijH\,dt}|\psi_0\rangle=\sum_{k=0}^{N-1}\gamma_k\underbrace{\sum_{j=-d}^d\widetilde{c}_je^{ijE_k\,dt}}_{\beta'_k}|E_k\rangle,
\end{equation}
where the last step follows by inserting \eqref{psi_0_decomp}.
Hence
\begin{equation}
\label{energy_error_numerator_1}
\begin{split}
    |\beta'_k|^2&=\left|\sum_{j=-d}^d\widetilde{c}_je^{ijE_k\,dt}\right|^2\\
    &=\left|\sum_{j=-d}^dc_je^{ijE_k\,dt}+\sum_{j=-d}^d(\widetilde{c}_j-c_j)e^{ijE_k\,dt}\right|^2\\
    &=\left|\beta_k+\sum_{j=-d}^d(\widetilde{c}_j-c_j)e^{ijE_k\,dt}\right|^2,
\end{split}
\end{equation}
where
\begin{equation}
    \beta_k\coloneqq p^*\left(\frac{\pi(E_k-E_0)}{R}\right),
\end{equation}
and thus $\beta_0=1$ and
\begin{equation}
\label{beta_bounds}
    |\beta_k|\le
    \begin{cases}
        1\quad\text{if $E_k-E_0<\delta$},\\
        2\left(1+\frac{\pi\delta}{R}\right)^{-d}\quad\text{if $E_k-E_0\ge\delta$}.
    \end{cases}
\end{equation}

To continue, it will be useful to introduce compact notations for the components of $\textbf{V}$: we decompose $\textbf{V}$ as
\begin{equation}
    \textbf{V}=
    \underbrace{
    \begin{bmatrix}
        |E_0\rangle&|E_1\rangle&\cdots&|E_{N-1}\rangle
    \end{bmatrix}
    }_{\coloneqq\mathbf{\Psi}}
    \Gamma
    \mathbf{\Phi},
\end{equation}
where $\Gamma\coloneqq\text{diag}(\gamma_0,\gamma_1,...,\gamma_{N-1})$, and $\mathbf{\Phi}$ is the matrix of phases from the time-evolutions of the energy eigenstates $|E_i\rangle$, given by
\begin{equation}
    \mathbf{\Phi}_{kj}\coloneqq e^{ijE_k\,dt}.
\end{equation}
Using this definition and continuing from \eqref{energy_error_numerator_1}, we have
\begin{equation}
\label{energy_error_numerator_2}
\begin{split}
    |\beta'_k|^2&=\left|\beta_k+\mathbf{\Phi}_k(\widetilde{c}-c)\right|^2\\
    &=\left|\beta_k+\mathbf{\Phi}_k(\underbrace{\Pi'-1}_{\coloneqq\Pi'^\perp})c\right|^2\\
    &=\left|\beta_k+\sum_{j=-d}^d[\mathbf{\Phi}\Pi'^\perp]_{kj}c_j\right|^2\\
    &\le2|\beta_k|^2+2\sum_{j=-d}^d|[\mathbf{\Phi}\Pi'^\perp]_{kj}|^2|c_j|^2\\
    &\le2|\beta_k|^2+\underbrace{2\sum_{j=-d}^d|[\mathbf{\Phi}\Pi'^\perp]_{kj}|^2}_{\coloneqq\alpha_k},
\end{split}
\end{equation}
where $\mathbf{\Phi}_k$ denotes the $k$th row of $\mathbf{\Phi}$ and the notation $[\cdot]_{kj}$ denotes $(k,j)$-th entry, and in the last step we again used the fact that $\|c\|\le1$.
Finally, we insert the bounds \eqref{beta_bounds} on $|\beta_k|$, yielding
\begin{equation}
    |\beta'_k|^2\le
    \begin{cases}
        2+\alpha_k\quad\text{if $E_k-E_0<\delta$},\\
        8\left(1+\frac{\pi\delta}{R}\right)^{-2d}+\alpha_k\quad\text{if $E_k-E_0\ge\delta$}.
    \end{cases}
\end{equation}

We now provide a collective upper bound on the $\alpha_k$:
\begin{equation}
\label{energy_error_numerator_3}
\begin{split}
    \sum_{k=0}^{N-1}|\gamma_k|^2\alpha_k&=2\sum_{k=0}^{N-1}\sum_{j=-d}^d|\gamma_k|^2|[\mathbf{\Phi}\Pi'^\perp]_{kj}|^2\\
    &=2\,\text{Tr}\left(\Pi'^\perp\mathbf{\Phi}^\dagger\Gamma^\dagger\Gamma\mathbf{\Phi}\Pi'^\perp\right).
\end{split}
\end{equation}
Now note that since $\mathbf{\Psi}^\dagger\mathbf{\Psi}=\mathds{1}_{D\times D}$,
\begin{equation}
    \textbf{S}=\textbf{V}^\dagger\textbf{V}=\mathbf{\Phi}^\dagger\Gamma^\dagger\mathbf{\Psi}^\dagger\mathbf{\Psi}\Gamma\mathbf{\Phi}=\mathbf{\Phi}^\dagger\Gamma^\dagger\Gamma\mathbf{\Phi},
\end{equation}
so \eqref{energy_error_numerator_3} becomes
\begin{equation}
\label{energy_error_numerator_4}
    \sum_{k=0}^{N-1}|\gamma_k|^2\alpha_k=2\,\text{Tr}\left(\Pi'^\perp\textbf{S}\Pi'^\perp\right)=2\,\sum_{i=0}^{D-1}\lambda_i\text{Tr}\left(\Pi'^\perp\Pi_i\Pi'^\perp\right)=2\,\sum_{i=0}^{D-1}\lambda_i\text{Tr}\left(\Pi_i\Pi'^\perp\Pi_i\right),
\end{equation}
where the second step follows by using idempotence of the middle projector and then the cyclic property of the trace; as a reminder, $\Pi'^\perp$ is the projector onto eigenspaces of $\textbf{S}'$ with eigenvalues below $\epsilon$, $\lambda_i$ are the eigenvalues of $\textbf{S}$ in weakly increasing order, and $\Pi_i$ are defined to be the corresponding spectral projectors.
This expression is a second important difference between the current proof and the proof of Theorem 3.1 in~\cite{epperly2021subspacediagonalization}, since that work the projector $\Pi'^\perp$ is replaced by a projector in the eigenbasis of $\textbf{S}$ itself, so the sum in \eqref{energy_error_numerator_4} terminates exactly at the largest value of $i$ such that $\lambda_i<\epsilon$.
Since in our case $\Pi'^\perp$ is a projector in the eigenbasis of $\textbf{S}'$, it is only approximately a projector in the eigenbasis of $\textbf{S}$, and we must upper bound all terms in the sum in \eqref{energy_error_numerator_4}.

The $\Pi_i$ are rank-one projectors, so we can further simplify to
\begin{equation}
    \sum_{k=0}^{N-1}|\gamma_k|^2\alpha_k=2\,\sum_{i=0}^{D-1}\lambda_i\|\Pi_i\Pi'^\perp\Pi_i\|.
\end{equation}
Next note that
\begin{equation}
    \|\Pi'^\perp\Pi_i\Pi'^\perp\|=\|\Pi_i\Pi'^\perp\|^2,
\end{equation}
by definition of $\|\Pi_i\Pi'^\perp\|$, so
\begin{equation}
    \sum_{k=0}^{N-1}|\gamma_k|^2\alpha_k=2\,\sum_{i=0}^{D-1}\lambda_i\|\Pi_i\Pi'^\perp\|^2\le2D\epsilon+2\,\sum_{i=I}^{D-1}(\lambda_i-\epsilon)\|\Pi_i\Pi'^\perp\|^2,
\end{equation}
where $I$ is the least integer such that $\lambda_I\ge\epsilon$.
Further define $J$ to be the least integer such that $\lambda_J\ge\epsilon+\|\textbf{S}'-\textbf{S}\|$:
\begin{equation}
    \sum_{k=0}^{N-1}|\gamma_k|^2\alpha_k\le2D\epsilon+2\,\sum_{i=I}^{J-1}(\lambda_i-\epsilon)\|\Pi_i\Pi'^\perp\|^2+2\,\sum_{i=J}^{D-1}(\lambda_i-\epsilon)\|\Pi_i\Pi'^\perp\|^2,
\end{equation}
Next we upper bound each $\|\Pi_i\Pi'^\perp\|^2$ in the last sum above using Davis-Kahan (\cref{davis_kahan}), yielding
\begin{equation}
\label{davis_kahan_application}
\begin{split}
    \sum_{k=0}^{N-1}|\gamma_k|^2\alpha_k&\le2D\epsilon+2\,\sum_{i=I}^{J-1}(\lambda_i-\epsilon)\|\Pi_i\Pi'^\perp\|^2+2\,\sum_{i=J}^{D-1}(\lambda_i-\epsilon)\frac{\|\textbf{S}'-\textbf{S}\|^2}{(\lambda_i-\epsilon)^2}\\
    &\le2D\epsilon+2(J-I)\|\textbf{S}'-\textbf{S}\|+2\|\textbf{S}'-\textbf{S}\|^2\sum_{i=J}^{D-1}\frac{1}{\lambda_i-\epsilon}\\
    &\le2D\epsilon+2(J-I)\|\textbf{S}'-\textbf{S}\|+2(D-J)\|\textbf{S}'-\textbf{S}\|\\
    &\le2D\left(\epsilon+\|\textbf{S}'-\textbf{S}\|\right),
\end{split}
\end{equation}
where the second step uses $\|\textbf{S}'-\textbf{S}\|$ to upper bound $\lambda_i-\epsilon$ with $i<J$, and the third step uses $\|\textbf{S}'-\textbf{S}\|$ to lower bound $\lambda_i-\epsilon$ with $i\ge J$, both by definition of $J$.

\end{proof}

\noindent
\textbf{Theorem~\ref{approx_projector_perturbed_H_thm}}.
\emph{
    Let $H$ and $H'$ be Hamiltonians.
    Let $E_0$ be the ground state energy of $H$, which we want to estimate.
    Let
    \begin{equation}
    \label{psi_def_app}
        |\psi\rangle=P|\psi_0\rangle
    \end{equation}
    be the approximately projected state defined in the statement of \cref{epperly_thm}.
    Let
    \begin{equation}
    \label{deltatilde_def_app}
        \widetilde{\Delta}\coloneqq E_1'-E_0
    \end{equation}
    be the gap between the ground state energy of $H$ and the first excited energy of $H'$, and let $\mathbf{1}$ denote the indicator function, i.e.,
    \begin{equation}
        \mathbf{1}(\delta'>\widetilde{\Delta})=
        \begin{cases}
            1\quad\text{if $\delta'>\widetilde{\Delta}$},\\
            0\quad\text{if $\delta'\le\widetilde{\Delta}$}.
        \end{cases}
    \end{equation}
    Then the error (as an estimate of $E_0$) of the expectation value of $H'$ with respect to $|\psi\rangle$ is upper bounded as
    \begin{equation}
    \label{approx_proj_bound_3_app}
        \frac{\langle\psi|(H'-E_0)|\psi\rangle}{\langle\psi|\psi\rangle}\le\delta'\mathbf{1}(\delta'>\widetilde{\Delta})+\|H'-H\|+6\|H\|\left(\frac{\|H'-H\|}{\delta'-\delta}+\frac{\zeta}{\||\psi\rangle\|^2}+\frac{8}{\||\psi\rangle\|^2}\left(1+\frac{\pi\delta}{R}\right)^{-2d}\right),
    \end{equation}
    where 
    \begin{equation}
    \label{zeta_def_app}
        \zeta\coloneqq2D\left(\epsilon+\|\textbf{S}\,'-\textbf{S}\,\|\right)
    \end{equation}
    and the bound holds for any parameters $0<\delta<\delta'<\|H\|$, provided
    \begin{equation}
    \label{thm3_assumption_app}
        \|H'-H\|<\delta'-\delta.
    \end{equation}
}
\begin{proof}

Let $\langle\cdot\rangle$ denote expectation value with respect to $|\psi\rangle$:
\begin{equation}
    \langle\hat{O}\rangle\coloneqq\frac{\langle\psi|\hat{O}|\psi\rangle}{\langle\psi|\psi\rangle}.
\end{equation}
Let $E_i$ and $E_i'$ be eigenvalues of $H$ and $H'$, respectively, in weakly increasing order.
Then if we define $\Pi'_K$ to be the spectral projector of $H'$ onto energies in a set or interval $K$ (and $\Pi_K$ similarly for $H$),
\begin{equation}
    \langle H'-E_0\rangle=\sum_{i=0}^{J-1}(E'_i-E_0)\langle\Pi_{\{E'_i\}}'\rangle=\sum_{i=0}^{I-1}(E'_i-E_0)\langle\Pi'_{\{E'_i\}}\rangle+\sum_{i=I}^{J-1}(E'_i-E_0)\langle\Pi'_{\{E'_i\}}\rangle,
\end{equation}
where $J$ is the number of distinct eigenvalues of $H'$ and $I$ is defined to be the least integer such that $E'_i\ge E_0+\delta'$.

To evaluate the first sum above, we must consider several cases.
First, if $E_0+\delta'\le E_0'$, then $I=0$ by definition and there are no terms in the first sum.
If that condition does not hold, but $E_0+\delta'\le E_1'$, then $I=1$ and the first sum is upper bounded by $E_0'-E_0\le\|H'-H\|$, with that inequality following by Weyl's theorem (\cref{weyl_thm}).
Finally, if $E_0+\delta'>E_1'$, then the first sum includes projectors onto excited states of $H'$ with energies up to $E_0+\delta'$, so it is upper bounded by $\delta'$ since that upper bounds each energy error in the corresponding low-energy subspace of $H'$.
Thus we obtain
\begin{equation}
\label{step_1}
\begin{split}
    \langle H'-E_0\rangle&\le\delta'\mathbf{1}(\delta'>E_1'-E_0)+\|H'-H\|\mathbf{1}(E_0'-E_0<\delta'\le E_1'-E_0)+\sum_{i=I}^{J-1}(E'_i-E_0)\langle\Pi'_{\{E'_i\}}\rangle\\
    &\le\delta'\mathbf{1}(\delta'>\widetilde{\Delta})+\|H'-H\|+(E'_{J-1}-E_0)\sum_{i=I}^{J-1}\langle\Pi'_{\{E'_i\}}\rangle\\
    &=\delta'\mathbf{1}(\delta'>\widetilde{\Delta})+\|H'-H\|+(E'_{J-1}-E_0)\langle\Pi'_{[E'_I,\infty)}\rangle,
\end{split}
\end{equation}
inserting the definition \eqref{deltatilde_def_app} of $\widetilde{\Delta}$ in the second step.

We now partition $|\psi\rangle$ into its components in the energy eigenspaces of $H$ above and below $E_0+\delta$:
\begin{equation}
\label{psi_partition}
    |\psi\rangle=|\psi_{\le E_0+\delta}\rangle+|\psi_{>E_0+\delta}\rangle.
\end{equation}
For $|\psi_{\le E_0+\delta}\rangle$, we simply upper bound its magnitude by the magnitude of $|\psi\rangle$.
By \eqref{agsp_def_part_1_app}, the square magnitude of $|\psi_{>E_0+\delta}\rangle$ is
\begin{equation}
\label{partition_bounds}
\begin{split}
    \||\psi_{>E_0+\delta}\rangle\|^2&=\sum_{E_k>E_0+\delta}|\gamma_k|^2|\beta'_k|^2\\
    &\le\sum_{E_k>E_0+\delta}|\gamma_k|^2\left(8\left(1+\frac{\pi\delta}{R}\right)^{-2d}+\alpha_k\right)\\
    &\le8\left(1+\frac{\pi\delta}{R}\right)^{-2d}\sum_{E_k>E_0+\delta}|\gamma_k|^2+\zeta\\
    &\le8\left(1+\frac{\pi\delta}{R}\right)^{-2d}+\zeta,
\end{split}
\end{equation}
where the second step uses \eqref{agsp_def_part_2_app} to fill in the $|\beta'_k|^2$, and the third step inserts \eqref{agsp_def_part_3_app} and then simplifies using \eqref{zeta_def_app}.

We now use the partition \eqref{psi_partition} to bound
\begin{equation}
\label{psi_high_energy_bound}
\begin{split}
    \langle\psi|\Pi'_{[E'_I,+\infty)}|\psi\rangle=\,&\langle\psi_{\le E_0+\delta}|\Pi'_{[E'_I,+\infty)}|\psi_{\le E_0+\delta}\rangle\\
    &+\langle\psi_{>E_0+\delta}|\Pi'_{[E'_I,+\infty)}|\psi_{\le E_0+\delta}\rangle\\
    &+\langle\psi_{\le E_0+\delta}|\Pi'_{[E'_I,+\infty)}|\psi_{>E_0+\delta}\rangle\\
    &+\langle\psi_{>E_0+\delta}|\Pi'_{[E'_I,+\infty)}|\psi_{>E_0+\delta}\rangle\\
    \le\,&\||\Pi'_{[E'_I,+\infty)}|\psi_{\le E_0+\delta}\rangle\|^2+2\||\Pi'_{[E'_I,+\infty)}|\psi_{\le E_0+\delta}\rangle\|\||\psi_{>E_0+\delta}\rangle\|+\||\psi_{>E_0+\delta}\rangle\|^2\\
    \le\,&2\||\Pi'_{[E'_I,+\infty)}|\psi_{\le E_0+\delta}\rangle\|^2+2\||\psi_{>E_0+\delta}\rangle\|^2,
\end{split}
\end{equation}
with the last step following by Young's inequality.
Eq.~\eqref{partition_bounds} already yields a bound for the second term inside the square, and we can bound the first term as follows:
\begin{equation}
\begin{split}
    \|\Pi'_{[E'_I,+\infty)}|\psi_{\le E_0+\delta}\rangle\|^2&=\langle\psi_{\le E_0+\delta}|\Pi'_{[E'_I,+\infty)}|\psi_{\le E_0+\delta}\rangle\\
    &=\langle\psi_{\le E_0+\delta}|\Pi_{(-\infty,E_0+\delta]}\Pi'_{[E'_I,+\infty)}\Pi_{(-\infty,E_0+\delta]}|\psi_{\le E_0+\delta}\rangle\\
    &\le\langle\psi_{\le E_0+\delta}|\psi_{\le E_0+\delta}\rangle\|\Pi_{(-\infty,E_0+\delta]}\Pi'_{[E'_I,+\infty)}\Pi_{(-\infty,E_0+\delta]}\|\\
    &\le\langle\psi|\psi\rangle\frac{\|H'-H\|}{E'_I-E_0-\delta},
\end{split}
\end{equation}
where the last step follows by Davis-Kahan (\cref{davis_kahan}) and because $\langle\psi_{\le E_0+\delta}|\psi_{\le E_0+\delta}\rangle\le\langle\psi|\psi\rangle$.
We also have by definition of $I$ that
\begin{equation}
    E'_I-E_0-\delta\ge E_0+\delta'-E_0-\delta=\delta'-\delta,
\end{equation}
so
\begin{equation}
    \|\Pi'_{[E'_I,+\infty)}|\psi_{\le E_0+\delta}\rangle\|^2\le\langle\psi|\psi\rangle\frac{\|H'-H\|}{\delta'-\delta}.
\end{equation}
Inserting this and \eqref{partition_bounds} into \eqref{psi_high_energy_bound} yields
\begin{equation}
\label{step_2}
    \langle\Pi'_{[E'_I,+\infty)}\rangle=\frac{\langle\psi|\Pi'_{[E'_I,+\infty)}|\psi\rangle}{\langle\psi|\psi\rangle}\le2\frac{\|H'-H\|}{\delta'-\delta}+\frac{2\zeta}{\langle\psi|\psi\rangle}+\frac{16}{\langle\psi|\psi\rangle}\left(1+\frac{\pi\delta}{R}\right)^{-2d}.
\end{equation}

Returning to \eqref{step_1}, we note that
\begin{equation}
\label{next_to_last}
    E_{J-1}'-E_0\le E_\text{max}-E_0+\|H'-H\|\le2\|H\|+\|H'-H\|<3\|H\|
\end{equation}
with the first step following by Weyl's Theorem (\cref{weyl_thm}) and the last following by \eqref{thm3_assumption_app}.
Note that the last step above does not change the leading order scaling, but if one wanted to obtain a bound with better constant factors, one could use the bound in the second-to-last step.
Inserting \eqref{step_2} and \eqref{next_to_last} into \eqref{step_1} yields our desired energy error bound \eqref{approx_proj_bound_3_app}.

\end{proof}

\noindent
\textbf{Theorem~\ref{complete_bound_thm}.}
\emph{
    Let $H$ be a Hamiltonian, let $(\textbf{H},\textbf{S}\,)=(\textbf{V}\,^\dagger H\textbf{V}, \textbf{V}\,^\dagger\textbf{V}\,)$ be a real-time Krylov matrix pair representing $H$ in the Krylov space $\text{span}(\textbf{V}\,)$, and let $(\textbf{H}\,',\textbf{S}\,')$ be a Hermitian approximation to $(\textbf{H},\textbf{S}\,)$.
    Let $E_0$ be the ground state energy of $H$, which we want to estimate.
    Let $\epsilon>0$ be a regularization threshold, and let
    \begin{equation}
    \label{chi_def_app}
        \chi\coloneqq\|\textbf{H}\,'-\textbf{H}\,\|+\|\textbf{S}\,'-\textbf{S}\,\|\|H\|
    \end{equation}
    be a measure of the noise.
    Let
    \begin{equation}
    \label{perturbed_overlap_app}
        |\gamma_0'|^2\coloneqq|\gamma_0|^2-2\epsilon-2\|\textbf{S}\,'-\textbf{S}\,\|
    \end{equation}
    be a noisy effective version of the initial state's overlap $|\gamma_0|^2$ with the true ground state.
    Let
    \begin{equation}
    \label{perturbed_gap_app}
        \Delta'\coloneqq\Delta-\frac{\chi}{|\gamma_0'|^2}
    \end{equation}
    be a noisy effective version of the spectral gap $\Delta$ of $H$.
    Then the lowest eigenvalue $\widetilde{E}_0$ of the thresholded matrix pair obtained from $(\textbf{H}\,',\textbf{S}\,')$ is bounded as
    \begin{equation}
    \label{approx_proj_bound_app}
        \widetilde{E}_0-E_0\le\delta'\mathbf{1}(\delta'>\Delta')+\frac{\chi}{|\gamma_0'|^2}+\frac{6\|H\|}{|\gamma_0'|^2}\left(\frac{\chi}{\delta'-\delta}+\zeta+8\left(1+\frac{\pi\delta}{2\|H\|}\right)^{-2d}\right)
    \end{equation}
    where $\zeta$ is defined in \eqref{zeta_def_app} and the bound holds for any parameters $\delta'>\delta>0$, provided the following assumptions hold:\\
    (i)
    \begin{equation}
    \label{thm4_assumption_app}
        \frac{\chi}{|\gamma_0'|^2}<\delta'-\delta,
    \end{equation}
    (ii)
    \begin{equation}
    \label{thm4_assumption_2_app}
        \epsilon\ge\|\textbf{S}\,'-\textbf{S}\|,
    \end{equation}
    and (iii) the right-hand side of \eqref{perturbed_overlap_app} is positive.
}

\begin{proof}

First, we note that using the first line in \eqref{psi_norm_lower_bound_app}, we can simplify \eqref{Hprime_diff} to
\begin{equation}
\label{Hprime_diff_pf}
    \|H'-H\|\le\frac{\left\|\textbf{H}'-\textbf{H}\right\|+\|H\|\left\|\textbf{S}'-\textbf{S}\right\|}{|\gamma_0|^2-\epsilon-\|\textbf{S}\,'-\textbf{S}\|}\le\frac{\chi}{|\gamma_0'|^2}.
\end{equation}
Inserting this into \eqref{approx_proj_bound_3_app} yields
\begin{equation}
\label{approx_proj_bound_4}
    \frac{\langle\psi|(H'-E_0)|\psi\rangle}{\langle\psi|\psi\rangle}\le\delta'\mathbf{1}(\delta'>E_1'-E_0)+\frac{\chi}{|\gamma_0'|^2}+6\|H\|\left(\frac{\chi}{(\delta'-\delta)|\gamma_0'|^2}+\frac{\zeta}{\||\psi\rangle\|^2}+\frac{8}{\||\psi\rangle\|^2}\left(1+\frac{\pi\delta}{R}\right)^{-2d}\right).
\end{equation}

Next, we note that the second lower bound on $\||\psi\rangle\|^2$ in \eqref{psi_norm_lower_bound_app} is exactly our definition \eqref{perturbed_overlap} of the effective overlap $|\gamma_0'|^2$, and by Weyl's Theorem (\cref{weyl_thm}), $E_1'-E_0$ is lower bounded as
\begin{equation}
    E_1'-E_0\ge E_1-\|H'-H\|-E_0=\Delta-\|H'-H\|\ge\Delta-\frac{\chi}{|\gamma_0'|^2}=\Delta',
\end{equation}
by definition \eqref{perturbed_gap_app}.
Replacing $\||\psi\rangle\|^2$ and $E_1'-E_0$ in \eqref{approx_proj_bound_4} with $|\gamma_0'|^2$ and $\Delta'$, respectively, to obtain a new upper bound yields
\begin{equation}
\label{approx_proj_bound_5}
    \frac{\langle\psi|(H'-E_0)|\psi\rangle}{\langle\psi|\psi\rangle}\le\delta'\mathbf{1}(\delta'>\Delta')+\frac{\chi}{|\gamma_0'|^2}+\frac{6\|H\|}{|\gamma_0'|^2}\left(\frac{\chi}{\delta'-\delta}+\zeta+8\left(1+\frac{\pi\delta}{R}\right)^{-2d}\right).
\end{equation}
These replacements are guaranteed to be well-defined by our assumption in the theorem statement that the right-hand side of \eqref{perturbed_overlap_app} is positive.

Finally, we note that by the Rayleigh-Ritz variational principle, the lowest energy $\widetilde{E}_0$ of the effective matrix pair $(\textbf{V}'^\dagger H'\textbf{V}',\textbf{V}'^\dagger\textbf{V}')$ is upper bounded as
\begin{equation}
    \widetilde{E}_0\le\frac{\langle\psi|H'|\psi\rangle}{\langle\psi|\psi\rangle}.
\end{equation}
Combining this with \eqref{approx_proj_bound_5}, we have
\begin{equation}
\label{approx_proj_bound_6}
\begin{split}
    \widetilde{E}_0-E_0\le\delta'\mathbf{1}(\delta'>\Delta')+\frac{\chi}{|\gamma_0'|^2}+\frac{6\|H\|}{|\gamma_0'|^2}\left(\frac{\chi}{\delta'-\delta}+\zeta+8\left(1+\frac{\pi\delta}{\|H\|}\right)^{-2d}\right),
\end{split}
\end{equation}
and using $2\|H\|$ to upper bound the spectral range $R$ of $H$, we obtain our desired result \eqref{approx_proj_bound_app}.

Note that replacing $\|H'-H\|$ in the assumption \eqref{thm3_assumption_app} of \cref{approx_projector_perturbed_H_thm} with $\chi/|\gamma_0'|^2$ yields the assumption \eqref{thm4_assumption_app} of the present theorem, which is stronger by \eqref{Hprime_diff_pf} and \eqref{psi_norm_lower_bound}.

\end{proof}

\end{widetext}
~

\end{document}